\PassOptionsToPackage{dvipsnames}{xcolor}
\documentclass[acmsmall,screen,nonacm]{acmart}

\usepackage{amsmath,amsfonts}
\usepackage{pifont}
\usepackage{array}
\usepackage{url}
\usepackage{graphicx}
\usepackage{enumitem}
\usepackage{wrapfig2}
\usepackage{adjustbox}

\usepackage{multirow}
\usepackage{booktabs}  
\usepackage{hyperref}  
\usepackage{subcaption}

\usepackage{xspace}
\usepackage{tikz}
\usetikzlibrary{positioning}
\usetikzlibrary{shapes}
\usetikzlibrary{decorations,arrows} 
\usepackage{forest}

\usepackage{tcolorbox}
\newtcolorbox[blend into=figures]{stlbox}[1][]{colframe=black!50,colback=white,#1}

\usepackage{calc}
\newlength\myheight
\newlength\mydepth
\settototalheight\myheight{Xygp}
\newtheorem{lemma}{Lemma}
\newtheorem{example}{Example}

\newcommand*{\set}[1]{\{ #1 \}}

\newcommand*{\marked}[1]{\overline{#1}}
\newcommand*{\stlunop}[3][]{\mathop{\mathsf{#2}^{#1}_{#3}}}
\newcommand*{\stlbinop}[5][]{{#4} \mathbin{{#2}^{#1}_{#3}} {#5}}
\newcommand*{\evenop}{\mathsf{F}}
\newcommand*{\globop}{\mathsf{G}}
\newcommand*{\untilop}{\mathsf{U}}
\newcommand*{\releaseop}{\mathsf{R}}
\newcommand*{\suntilop}{\mathsf{sU}}
\newcommand*{\sreleaseop}{\mathsf{sR}}
\newcommand*{\until}[3]{\stlbinop{\untilop}{#1}{#2}{#3}}
\newcommand*{\release}[3]{\stlbinop{\releaseop}{#1}{#2}{#3}}
\newcommand*{\even}[2][]{\stlunop[{#1}]{\evenop}{#2}}
\newcommand*{\glob}[2][]{\stlunop[{#1}]{\globop}{#2}}
\newcommand*{\meven}[2][]{\stlunop[{#1}]{\marked{\evenop}}{#2}}
\newcommand*{\mglob}[2][]{\stlunop[{#1}]{\marked{\globop}}{#2}}
\newcommand*{\suntil}[4][]{\stlbinop[#1]{\suntilop}{#2}{#3}{#4}}
\newcommand*{\srelease}[4][]{\stlbinop[#1]{\sreleaseop}{#2}{#3}{#4}}
\newcommand*{\anybinop}[4][]{\stlbinop[{#1}]{\mathsf{B}}{#2}{#3}{#4}}
\newcommand*{\msuntil}[4][]{\stlbinop[#1]{\marked{\suntilop}}{#2}{#3}{#4}}
\newcommand*{\msrelease}[4][]{\stlbinop[#1]{\marked{\sreleaseop}}{#2}{#3}{#4}}
\newcommand*{\anyunop}[2][]{\stlunop[#1]{A}{#2}}

\renewcommand{\implies}{\rightarrow}
\renewcommand{\iff}{\leftrightarrow}

\renewcommand*{\exp}{\operatorname{tex}}
\DeclareMathOperator{\pcl}{pcl}
\DeclareMathOperator{\horiz}{h}
\DeclareMathOperator{\unmark}{unmark}

\newcommand{\cmark}{\textcolor{ForestGreen}{\ding{51}}}
\newcommand{\xmark}{\textcolor{red}{\ding{55}}}

\begin{document}

\title{A Tree-Shaped Tableau for Checking the Satisfiability of Signal Temporal Logic with Bounded Temporal Operators}

\author{Beatrice Melani}
\affiliation{%
  \institution{Politecnico di Milano}
  \city{Milan}
  \country{Italy}}
\orcid{0000-0002-9051-1597}
\email{beatrice.melani@polimi.it}

\author{Ezio Bartocci}
\orcid{0000-0002-8004-6601}
\email{ezio.bartocci@tuwien.ac.at}

\author{Michele Chiari}
\affiliation{%
  \institution{TU Wien}
  \city{Vienna}
  \country{Austria}}
\orcid{0000-0001-7742-9233}
\email{michele.chiari@tuwien.ac.at}


\begin{abstract}
Signal Temporal Logic (STL) is a widely recognized formal specification language to express rigorous temporal requirements on mixed analog signals produced by cyber-physical systems (CPS).
A relevant problem in CPS design is how to efficiently and automatically check whether a set of STL requirements is logically consistent.
This problem reduces to solving the STL satisfiability problem, which is decidable when we assume that our system operates in discrete time steps dictated by an embedded system's clock.

This paper introduces a novel tree-shaped, one-pass tableau method for satisfiability checking of discrete-time STL with bounded temporal operators. Originally designed to prove the consistency of a given set of STL requirements, this method has a wide range of applications beyond consistency checking.
These include synthesizing example signals that satisfy the given requirements, as well as verifying or refuting the equivalence and implications of STL formulas.

Our tableau exploits redundancy arising from large time intervals in STL formulas to speed up satisfiability checking,
and can also be employed to check Mission-Time Linear Temporal Logic (MLTL) satisfiability.
We compare our tableau with Satisfiability Modulo Theories (SMT) and First-Order Logic encodings from the literature
on a benchmark suite, partly collected from the literature, and partly provided by an industrial partner.
Our experiments show that, in many cases, our tableau outperforms state-of-the-art encodings.
\end{abstract}


\maketitle

\section{Introduction}
\label{sec:intro}

The formal specification of requirements is an essential engineering activity~\cite{BoufaiedJBBP21} throughout the development process of a cyber-physical system (CPS), from its design and testing/debugging~\cite{BoufaiedMBB23,cps2025,BartocciMMMN21} before deployment, to monitoring and runtime verification after deployment~\cite{JaksicBGKNN15}. 

Real-time temporal logics~\cite{Koymans90} such as Signal Temporal Logic~\cite{MalerN04} (STL) provide a very well-established~\cite{BartocciDDFMNS18} and intuitive formal framework to express and reason about temporal requirements over discrete and continuous signals.  These mixed-analog signals characterize both the external behavior of the physical environment that the CPS measures through its sensors and the internal behavior of the embedded computational device controlling the physical environment.  STL has become widely used in various industrial settings, as evidenced by numerous scientific papers on its applications published by industrial research laboratories~\cite{MolinANZBE23,Kapinski2016STLibAL,RoehmHM17}. 

The process of requirement engineering is highly complex and heterogeneous: requirements are collected and documented at all stages of the CPS design process, initially in natural language, which is then manually translated into unambiguous formal specifications. This step requires significant skills in formal methods on the part of human designer.
More recently, transformer-based machine learning models~\cite{HeBNIG22} such as large language models (LLMs), have also been employed to assist the engineer in translating requirements from natural language to STL or other temporal logics, but this task also requires accurate human supervision~\cite{MohammadinejadPXKTD24}, since LLMs are prone to errors and hallucinations. Furthermore, guessing requirements about the physical environment in which a CPS operates
can be very cumbersome when done manually. These formal specifications can be inferred from observed data using automated techniques such as specification mining methods~\cite{BartocciMNN22}. 

In this scenario, \emph{an important open problem is how to validate efficiently a set of STL requirements and make sure that they are consistent}.
The definition of consistency of a requirement set that is adopted in the context of this paper is the same as the ISO/IEC/IEEE 2918 standard \cite{ISO29148} in which a requirement set is considered consistent when it ``does not have individual requirements which are contradictory''.

\paragraph{An example}
Consider the following two requirements of a railroad system~\cite{BaeLee19}:

\begin{tcolorbox}[
    colframe=black!50, colback=white, 
    title={STL Requirements for a Railroad system from~\cite{BaeLee19}}
]
\renewcommand{\arraystretch}{1.5} 
\footnotesize
\begin{tabular}{>{\raggedright}p{5.8cm} | p{6.6cm}}
    
    $\glob{[3,50]}(\even{[5,20]}(a \geq 80))$
    &
    For all time steps in \([3,50]\), there must be a future moment within \([5,20]\) where the angle $a$ reaches at least the value 80. \\ 
    
    $\glob{[10,60]}(a \geq 80 \implies \glob{[20,40]}(a < 60))$ 
    &
    If at any time in \([10,60]\) the angle is at least 80, then for all time steps in the following interval \([20,40]\), the angle must remain below 60. \\ 
\end{tabular}
\end{tcolorbox}

These two requirements are inconsistent, because there exists no signal that would satisfy these two STL formulas in conjunction. However, it is very hard for an engineer to understand at a glance that they are inconsistent.
Proving it manually requires several logical steps and a good knowledge of formal methods.
The first requirement is equivalent to another formula:
\begin{equation}
\label{eq:railroad-shift-interval}
\glob{[3,50]}(\even{[5,20]}(a \geq 80)) \leftrightarrow \glob{[8,55]}(\even{[0,15]}(a \geq 80))
\end{equation}
We can transform the implication in the second requirement by using $(\phi \rightarrow \psi) \leftrightarrow (\neg \phi \vee \psi)$ and then we need to prove the consistency of this expression:
\[
\glob{[8,55]}(\even{[0,15]}(a \geq 80))  \wedge \glob{[10,60]}( (a < 80) \vee \glob{[20,40]}(a < 60))
\]
If we consider STL interpreted over discrete-time signals we can rewrite the formula as follows:
\[
\glob{[8,29]}(\even{[0,15]}(a \geq 80)) \wedge \glob{[30,55]}(\even{[0,15]}(a \geq 80)) \wedge \glob{[10,60]}( (a < 80) \vee \glob{[20,40]}(a < 60))
\]
Using the rule  $\glob{[i+k,j+k]}(\even{[l,m]} \phi) \leftrightarrow \glob{[i,j]} (\even{[l+k,m+k]} \phi)$ s.t. ($i,j,l,m,k \in \mathbb{N}, i < j, l < m$) we can obtain the following formula:
\begin{equation}
\glob{[8,29]}(\even{[0,15]}(a \geq 80)) \wedge \glob{[10,25]}(\textcolor{Green}{\even{[20,35]}(a \geq 80)}) \wedge \textcolor{blue}{\glob{[10,60]}( (a < 80) \vee \glob{[20,40]}(a < 60))}
\label{eq:railroad-shift-again}
\end{equation}
We can split the third globally expression in \textcolor{blue}{blue} in the conjunction as:
\begin{equation}
\glob{[10,25]}( (a < 80) \vee \textcolor{Green}{\glob{[20,40]}(a < 60)}) \wedge \glob{[26,60]}( (a < 80) \vee \glob{[20,40]}(a < 60))
\label{eq:railroad-split-glob}
\end{equation}
The two formulas in \textcolor{Green}{green} within $\glob{[10,25]}$ in \eqref{eq:railroad-shift-again} and \eqref{eq:railroad-split-glob} are contradictory, so we can exclude $\glob{[20,40]}(a < 60)$, and only $\glob{[10,25]}(a < 80)$ in \eqref{eq:railroad-split-glob} may hold.
Using the rule $\glob{[i,j]}(\phi) \wedge \glob{[i,j]}(\psi) \leftrightarrow \glob{[i,j]}(\phi \wedge \psi)$ s.t. ($i,j \in \mathbb{N}, i < j$) and the distributivity of $\wedge$, we merge the two $\glob{[10,25]}$ operators:
\[
\glob{[8,29]}(\even{[0,15]}(a \geq 80)) \wedge \textcolor{purple}{\glob{[10,25]}(\even{[20,35]}(a \geq 80) \wedge (a < 80) )} \wedge \glob{[26,60]}( (a < 80) \vee \glob{[20,40]}(a < 60))
\]
Using the rule $\glob{[i,j]}(\phi \wedge \psi) \leftrightarrow \glob{[i,j]}(\phi) \wedge \glob{[i,j]}(\psi) $ we then split the globally expression in purple:
\[\textcolor{purple}{\glob{[10,25]}(\even{[20,35]}(a \geq 80)) \wedge \glob{[10,25]}(a < 80)}\]


Now, to prove that the original requirement is contradictory, we just need to prove that this formula is always false:
\[\glob{[8,29]}(\even{[0,15]}(a \geq 80)) \wedge \glob{[10,25]}(a < 80)\]
By applying the following rule $\glob{[i+1,j]} (\even{[l,m]} \phi) \leftrightarrow \even{[l+i+1,m+i+1]} (\phi) \wedge \glob{[i,j]}(\even{[l,m]} \phi)   $  three times:
\[\glob{[11,29]}(\even{[0,15]}(a \geq 80)) \wedge \even{[8,23]}(a \geq 80) \wedge \even{[9,24]}(a \geq 80) \wedge \textcolor{red}{\even{[10,25]}(a \geq 80) \wedge \glob{[10,25]}(a < 80)}\]
The expression in red is always false because the value of $a$ is required to be at least greater than or equal to 80 and, at the same time, always less than 80 in the time interval $[10,25]$.
This contradiction makes the original requirement false for each possible signal of $a$.

\paragraph{Consistency checking}
Consistency checking is a relevant issue in the context of requirement engineering. In fact, many software errors are introduced during the requirement phase of a project as a result of poor validation of the set of requirements.  Fixing such errors, especially when they are detected late in the development cycle of the project, is difficult and costly. Furthermore, there is evidence that errors in functional and interface requirements are often the main source of critical safety software errors that could potentially cause serious accidents~\cite{Heitmeyer96}. 
When the set of requirements becomes large and complex, this task becomes impossible to perform manually.  Verifying automatically whether a set of requirements is consistent amounts to checking the satisfiability of the requirement set, but it poses a challenge, especially when dealing with large-scale CPS.
To make the problem of STL satisfiability feasible, without losing the level of expressiveness needed to describe real cyber-physical systems, we focus our
work on bounded-time discrete-time STL. Although these limitations may seem restrictive, they  are not significant on a practical level.  Indeed, a finite time-horizon is an acceptable constraint, provided that the chosen horizon is long enough to assess the satisfiability of the requirements in a real-world scenario (i.e., the flight time of an aircraft, the length of the mission of a drone). As for the limitation to discrete-time STL, it is consistent with the discretization of continuous behaviors that occurs when using digital hardware.

A common approach to solving the satisfability problem for discrete-time STL is to encode the formula as a satisfiability modulo theories (SMT) problem.
For example, solving the satisfiability problem for the formula $\glob{[2,3]}(\even{[0,1]}(a \geq 80))$ is equivalent to solving this SMT problem instance
$((a_2 \geq 80) \vee (a_3 \geq 80)) \wedge ((a_3 \geq 80) \vee (a_4 \geq 80))$.
However, a clear issue of this approach is that it requires a number of variables that grows with the length of the considered time intervals (i.e., the formula horizon) and time step in the temporal operators.

\paragraph{Our approach} In this paper we consider a novel alternative approach.  We solve the consistency problem for a set of discrete-time STL requirements with bounded-time temporal operators by constructing a tree-shaped, one-pass tableau. The idea behind our tableau is that in order to prove the consistency of a (discrete-time) requirement set within a finite time horizon, one does not need to check the consistency of requirements at every time instant. Indeed, requirements with non-intersecting time domains are trivially consistent with each other and for requirements with intersecting time domains the check can be performed only for those time instants in which a new constraint is imposed, while the time instants in between can be skipped. In practice, when in consecutive nodes of the tableau the propositional part of the temporal formula extracted by the decomposition of the temporal operators is the same, the evaluation of the consistency is performed only once.
Given the structure of the algorithm, which allows for temporal jumps, the method is particularly suitable to check the consistency of requirements that need to be upheld for a long time horizon (e.g., the flight time of an aircraft).

\paragraph{Our contribution} We summarize the contributions of this paper as follows:
\begin{itemize}
\item We introduce a novel method, based on a tree-shaped, one-pass tableau, to decide the satisfiability problem for discrete-time STL with bounded temporal operators. This method can be applied in different applications, from checking the consistency of a set of CPS requirements to generating examples of signals that satisfy a formula or improving specification mining techniques~\cite{BartocciMNN22}.  We prove both the soundness and the completeness of our method.
\item We provide several heuristics and optimizations to considerably improve the efficiency of our approach in practice. These heuristics take advantage of several decomposition properties of bounded temporal operators such as Globally and Finally in the discrete-time setting.
\item We implement our approach in a prototype tool written in Python. We successfully applied our method to solve the consistency checking problem on several benchmarks, including a realistic set of requirements provided by one of our industry partners.
\item We evaluated our method by comparing it against an SMT-based approach, using SMT solvers for satisfiability checking.  We also compare our method with another approach proposed in the literature to check the satisfiability of the requirements
in Mission-time LTL~\cite{LiVR22}, a ``bounded variant of Metric Temporal Logic (MTL) over naturals''. To compare with their approach, we restrict our language to accept only Boolean predicates. 
\end{itemize}

\paragraph{Paper organization}
The structure of the paper is as follows:
Sec. \ref{sec:rw} presents a review of the literature of relevant works on the proposed topic, while Sec. \ref{sec:stl} briefly presents STL.
Sec. \ref{sec:tableau} describes the tableau method first through some examples to illustrate the idea behind the algorithm, and then with a formal definition. Proofs of its correctness, completeness, and complexity are also included.
Sec. \ref{sec:heur} summarizes the heuristics and optimizations devised to make the algorithm more efficient. Sec. \ref{sec:exp} contains the experimental evaluation of the tableau against different benchmarks, while Sec. \ref{sec:ablation} presents an ablation analysis of the proposed optimizations. The last section outlines conclusions and future work.

\section{Related Work}
\label{sec:rw}
Signal Temporal Logic (STL)~\cite{MalerN04} is a formal language that has become increasingly popular due to its expressiveness in describing properties of continuous signals.
STL extends Metric Interval Temporal Logic (MITL) restricted to temporal operators in bounded non-singleton time intervals ($MITL_{[a,b]}$)~\cite{AlurFH96}, enriching it with predicates over real-valued signals.
Despite being widely used in the context of cyber-physical systems, STL still poses interesting challenges, in particular solving the satisfiability problem, which tries to determine whether there exists a set of signals that satisfy a given specification.
Deciding satisfiability for STL is equivalent to deciding the satisfiability of ($MITL_{[a,b]}$) and is EXPSPACE-complete~\cite{AlurH93}.
Despite its complexity, solving the STL satisfiability problem is instrumental in many domains beyond consistency check, such as model predictive control~\cite{RamanDMMSS14} and automatic trace generation~\cite{PrabhakarLK18}. Since STL is an extension of a fragment of MITL, a related relevant work to mention is the SMT encoding proposed in~\cite{BersaniRP15,BersaniRP16} to solve the satisfiability problem for different continuous-time temporal logics such as MITL and Quantified Temporal Logic (QTL)~\cite{HirshfeldR05}.  In contrast with these works, our satisfiability checking procedure supports discrete-time STL with Boolean predicates and inequalities over continuous values (which are not included in MITL) as basic propositions.

\citet{Dokhanchi} present a framework for debugging STL specifications
covering validity, vacuity, and redundancy checking.
They translate STL formulas into MITL by replacing real constraints with Boolean propositions through a brute-force approach.
Then, they translate MITL formulas into CLTL-over-clocks and use qtlSolver and Zot~\cite{BersaniRP13} to check satisfiability.
They show that satisfiability of a restricted fragment of MITL is reducible to LTL satisfiability, which can be solved faster with existing tools~\cite{CimattiCGGPRST02,RozierV10}.
While we do not consider vacuity and redundancy in this paper,
our tool for STL satisfiability checking could be used as a backend for the algorithms by \citet{Dokhanchi}.
Our approach has the advantage that it does not need the costly preprocessing to translate STL formulas into MITL.
\citet{RamanDMMSS14} present an SMT-based encoding of discrete-time STL to address the control synthesis problem in model predictive control under an STL specification.
In Sec.~\ref{sec:exp}, we compare our tree-shaped tableau approach with their SMT-based encoding.

\citet{LiVR22} address the satisfiability problem for requirements expressed in Mission-time LTL (MLTL)~\cite{ReinbacherRS14}.
Although MLTL specifications can also be encoded in discrete-time STL, the two logics are not equivalent in expressivity because MLTL supports only Boolean signals.
\citet{LiVR22} propose three methods for satisfiability checking:
the first converts the MLTL formula into LTL or LTL-on-finite-traces (LTLf), and exploits existing solvers for these logics;
the second translates the MLTL formula into an SMV~\cite{0071856} Boolean transition system to be checked with the NuXmv~\cite{CavadaCDGMMMRT14} model checker;
the third encodes the MLTL formula into a first-order logic (FOL) formula and solves it using the Z3~\cite{Z3} Theorem Prover.
Their results show that the second approach performs well for formulas with small temporal intervals,
while the third scales better for larger intervals. 
 Our approach, in contrast, offers an alternative by using a tableau method that can skip time instants when they are not necessary to prove the inconsistencies, which can provide benefits in performance.
 In the experimental section, we compare our approach with the FOL encoding proposed by~\citet{LiVR22} using their benchmarks.  
 
 Other approaches~\cite{BaeLee19,LeeYB21,roehm} demonstrate how to perform STL bounded model checking of hybrid/dynamical systems modeling cyber-physical systems (CPS) as solving, using SMT-solvers, a satisfiability problem of ``a first-order logic formula over reals''~\cite{LeeYB21}.
 \citet{LeeYB21} also consider the problem of satisfiability of STL for the continuous-time setting.
 In contrast, our approach considers only the discrete-time interpretation of STL.


Tableau methods have been used extensively to assess the satisfiability of Linear Temporal Logic (LTL)~\cite{Pnueli77} specifications, most of them being graph-shaped and two-pass~\cite{Wolp,Licht}
(i.e., first a graph representing all the paths of the transition system describing the formula is constructed and then such graph is traversed once more to find a model for the input formula).
A tree-shaped, one-pass tableau for LTL has been developed by \citet{Rey},
and extended to LTL with past by \citet{GiganteMR17}.
This kind of tableau has the advantage of simultaneously building a branch of the tree and deciding whether it is an accepting branch.
Reynold's tableau has been implemented by \citet{BertelloGMR16} and later parallelized \cite{McCabeDanstedR17}.
A one-pass, tree-shaped tableau for another real-time temporal logic,
Timed Propositional Temporal Logic (TPTL)~\cite{AlurH94}, has been proposed by \citet{GeattiGMR21}.
More recently, \citet{GeattiGMV24} devised a SAT encoding of Reynold's tableau to check the satisfiability of LTL and LTL interpreted on finite traces.

Our method takes inspiration from \citet{Rey} to design a single-pass, tree-shaped tableau for STL.
While Reynolds' \textsf{PRUNE} rule exploits node repetition to prune redundant tree branches,
our tableau exploits it to skip redundant portions of the tree arising from large time intervals.
Additionally, our tableau employs an SMT solver to check satisfiability of linear constraints in temporal formulas.
To the best of our knowledge, there are no other tableau methods to assess the satisfiability of discrete-time STL with bounded-time operators.

\section{Signal Temporal Logic}
\label{sec:stl}

STL~\cite{MalerN04} is defined on a temporal domain $\mathbb{T}$, which is usually real numbers.
In this paper, however, we consider $\mathbb{T} = \mathbb{N}$.
Let $S = \{ x_1, \dots, x_n \}$ be a set of signal variables:
a \emph{signal} is a function $w : \mathbb{T} \rightarrow \mathbb{R}^n$.
We define the projection $w_R$ of a signal $w$ on a set of variables $R \subseteq S$ in the usual way.
STL formulas follow the syntax:
\[
\varphi := \top \mid f(R) = k \mid f(R) > k \mid \neg \varphi \mid \varphi_1 \lor \varphi_2
  \mid \until{I}{\varphi_1}{\varphi_2}
\]
where $R \subseteq S$ is any set of signal variables,
$f : \mathbb{R}^{|R|} \rightarrow \mathbb{R}$ any linear function of signal values,
$k \in \mathbb{Q}$ any rational constant,
and $I = [a,b]$ with $a, b \in \mathbb{T}$.

The semantics of STL are given with respect to a signal $w$ and a time instant $t \in \mathbb{T}$ as follows:
\begin{align*}
&(w,t) \models \top && \text{(always true)} \\
&(w,t) \models f(R) = k     & \text{iff } & f(w_R(t)) = k \\
&(w,t) \models f(R) > k     & \text{iff } & f(w_R(t)) > k \\
&(w,t) \models \neg \varphi & \text{iff } & (w,t) \not\models  \varphi \\
&(w,t) \models \varphi_1 \lor \varphi_2 & \text{iff } & (w,t) \models \varphi_1 \text{ or } (w,t) \models \varphi_2 \\
&(w,t) \models \until{[a,b]}{\varphi_1}{\varphi_2}  & \text{iff } & \exists t' \in [t+a, t+b] : (w, t') \models \varphi_2
\text{ and } \forall t'' \in [t, t'] : (w,t'') \models \varphi_1
\end{align*}
Note that it is possible to encode \emph{atomic propositions} as signals by defining a variable $x_p$
for each proposition $p$, and asserting $p$ as $x_p = 1$.

We employ derived propositional operators ($\land$, $\implies$)
and derived temporal operators \emph{eventually} (or \emph{finally}) $\even{I}$
and \emph{always} (or \emph{globally}) $\glob{I}$ with the usual semantics.
We define the \emph{release} operator $\releaseop$,
and two \emph{strict} variants of the until ($\suntilop$) and release ($\sreleaseop$) operators with the following semantics:
\begin{align*}
&(w,t) \models \release{[a,b]}{\varphi_1}{\varphi_2} & \text{iff } & \forall t' \in [t+a, t+b] : (w, t') \models \varphi_2 
\text{ or } \exists t'' \in [t, t'] : (w,t'') \models \varphi_1 \\
&(w,t) \models \suntil{[a,b]}{\varphi_1}{\varphi_2} & \text{iff } & \exists t' \in [t+a, t+b] : (w, t') \models \varphi_2 
\text{ and } \forall t'' \in [t+a, t'-1] : (w,t'') \models \varphi_1 \\
&(w,t) \models \srelease{[a,b]}{\varphi_1}{\varphi_2} & \text{iff } & \forall t' \in [t+a, t+b] : (w, t') \models \varphi_2 
\text{ or } \exists t'' \in [t+a, t'-1] : (w,t'') \models \varphi_1
\end{align*}

A formula $\phi$ is in \emph{strict normal form} if it only contains
the $\lor$, $\land$, $\suntilop$, and $\sreleaseop$ operators in positive form,
and possibly negated terms and $\top$.
Any formula can be transformed into one in strict normal form by applying the following substitutions recursively:
\begin{align*}
&\begin{aligned}
&\neg (\varphi_1 \lor \varphi_2) \equiv \neg \varphi_1 \land \neg \varphi_2 &
&\neg (\varphi_1 \land \varphi_2) \equiv \neg \varphi_1 \lor \neg \varphi_2 \\
&\neg (\suntil{I}{\varphi_1}{\varphi_2}) \equiv \srelease{I}{\neg \varphi_1}{\neg \varphi_2} &
&\neg (\srelease{I}{\varphi_1}{\varphi_2}) \equiv \suntil{I}{\neg \varphi_1}{\neg \varphi_2} \\
&\even{I} \varphi \equiv \suntil{I}{\top}{\varphi} &
&\glob{I} \varphi \equiv \neg \even{I} \neg \varphi \\
\end{aligned} \\
&\until{[a,b]}{\varphi_1}{\varphi_2} \equiv \glob{[0,a]} \varphi_1 \land \suntil{[a,b]}{\varphi_1}{(\varphi_1 \land \varphi_2)} \\
&\release{[a,b]}{\varphi_1}{\varphi_2} \equiv \even{[0,a]} \varphi_1 \lor \srelease{[a,b]}{(\varphi_1 \land \varphi_2)}{\varphi_2}
\end{align*}

Given any formula $\varphi = \anyunop{[a,b]} \varphi_1$ with $\mathsf{A} \in \{\mathsf{F}, \mathsf{G}\}$,
or $\varphi = \anybinop{[a,b]}{\varphi_1}{\varphi_2}$ with $\mathsf{B} \in \{\untilop, \releaseop, \suntilop, \sreleaseop\}$,
we define $I(\varphi) = [a,b]$, $I_\ell(\varphi) = a$ and $I_u(\varphi) = b$.

Given a formula $\varphi$, we define its propositional closure $\pcl(\varphi)$ as follows:
\begin{itemize}
    \item if $\varphi = \neg \varphi_1$, then $\pcl(\varphi) = \{\varphi_1\}$;
    \item if $\varphi = \varphi_1 \land \varphi_2$ or $\varphi = \varphi_1 \lor \varphi_2$, then $\pcl(\varphi) = \{\varphi_1, \varphi_2\}$;
    \item if $\varphi = \anyunop{I} \varphi_1$ with $\mathsf{A} \in \{\mathsf{F}, \mathsf{G}\}$,
    or $\varphi = \anybinop{I}{\varphi_1}{\varphi_2}$ with $\mathsf{B} \in \{\untilop, \releaseop, \suntilop, \sreleaseop\}$,
    then $\pcl(\varphi) = \{\varphi\}$.
\end{itemize}

\section{Tree-Shaped Tableau}
\label{sec:tableau}

We introduce the tableau in two phases:
in Sec.~\ref{sec:basic-tableau} a \emph{basic} tableau,
that is theoretically sound but slow;
in Sec.~\ref{sec:jump} an optimization that, as we shall see in Sec.~\ref{sec:exp},
improves performances substantially.

\subsection{The Basic Tableau}
\label{sec:basic-tableau}

We present a one-pass, tree-shaped tableau for STL satisfiability on a discrete time domain.
A tableau is a tree in which each node $u$ is labeled with a set of formulas $\Gamma(u)$, and a time counter $t(u)$.
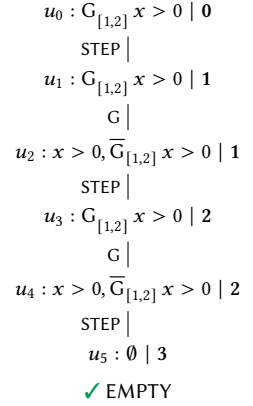
\begin{wrapfigure}[18]{r}
\footnotesize
\begin{forest}
  for tree={
    myleaf/.style={label=below:{\strut#1}}
  },
  [{$u_0: \glob{[1,2]} x > 0 \mid \mathbf{0}$}
    [{$u_1: \glob{[1,2]} x > 0 \mid \mathbf{1}$},
     edge label={node[midway,left,font=\scriptsize\sffamily]{STEP}},
      [{$u_2: x > 0, \mglob{[1,2]} x > 0 \mid \mathbf{1}$},
       edge label={node[midway,left,font=\scriptsize\sffamily]{G}},
        [{$u_3: \glob{[1,2]} x > 0 \mid \mathbf{2}$},
          edge label={node[midway,left,font=\scriptsize\sffamily]{STEP}},
         [{$u_4: x > 0, \mglob{[1,2]} x > 0 \mid \mathbf{2}$},
           edge label={node[midway,left,font=\scriptsize\sffamily]{G}},
           [{$u_5: \emptyset \mid \mathbf{3}$},
             edge label={node[midway,left,font=\scriptsize\sffamily]{STEP}},
             myleaf={\cmark\ \textsf{EMPTY}}]
         ]
        ]
      ]
    ]
  ]
\end{forest}
\caption{Basic tableau for $\glob{[1,2]} x > 0$.
    Each node's time counter is shown in bold face.}
\label{fig:G-example}
\end{wrapfigure}
The root $u_0$ is labeled with a formula $\phi$ to be checked for satisfiability.
The tree is generated by iteratively applying to each node a set of rules that generate its children,
until either a branch is \emph{accepted}, proving the existence of a model for $\phi$%
---i.e., a signal that satisfies $\phi$---%
or all branches are \emph{rejected}, proving that $\varphi$ is unsatisfiable.

\subsubsection{Examples}
We first illustrate informally the tableau built for formula $\phi = \glob{[1,2]} x > 0$,
shown in Fig.~\ref{fig:G-example}.
The root node $u_0$ is labeled with $\phi$, with time counter $t(u_0) = 0$.
Since the lower bound of the $\globop$ operator is 1,
there is nothing to do yet, and we apply the \textsf{STEP} rule to proceed with
the next time instant of the satisfying signal that we are trying to build.
The new node $u_1$ has time $t(u_1) = 1$.
For $\phi$ to be satisfied, $x > 0$ must hold at $t = 1$.
We enforce this constraint by applying an \emph{expansion} rule,
which adds a child $u_2$ to $u_1$, still with time $t(u_2) = 1$:
$u_2$ is labeled with $x > 0$, but also with $\phi$, because it still has not been satisfied.
$\phi$ is, however, \emph{marked} ($\marked{\globop}$),
to state that more obligations have to be met in the future for it to be satisfied.

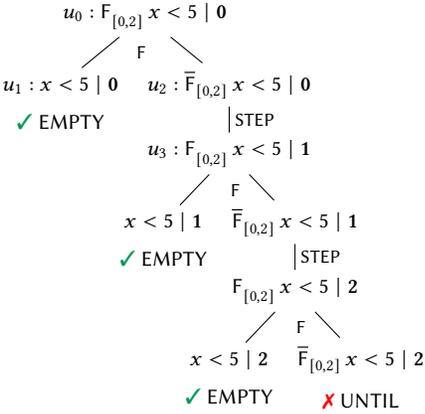
\begin{wrapfigure}[17]{l}
\footnotesize
\begin{forest}
  for tree={
    myleaf/.style={label=below:{\strut#1}}
  },
  [{$u_0: \even{[0,2]} x < 5 \mid \mathbf{0}$}
    [{$u_1: x < 5 \mid \mathbf{0}$}, myleaf={\cmark\ \textsf{EMPTY}},
     edge label={node[midway,left,xshift=20pt,font=\scriptsize\sffamily]{F}}
    ],
    [{$u_2: \meven{[0,2]} x < 5 \mid \mathbf{0}$},
      [{$u_3: \even{[0,2]} x < 5 \mid \mathbf{1}$},
        edge label={node[midway,left,xshift=20pt,font=\scriptsize\sffamily]{STEP}},
        [{$x < 5 \mid \mathbf{1}$}, myleaf={\cmark\ \textsf{EMPTY}},
          edge label={node[midway,left,xshift=20pt,font=\scriptsize\sffamily]{F}}
        ],
        [{$\meven{[0,2]} x < 5 \mid \mathbf{1}$},
          [{$\even{[0,2]} x < 5 \mid \mathbf{2}$},
            edge label={node[midway,left,xshift=20pt,font=\scriptsize\sffamily]{STEP}},
            [{$x < 5 \mid \mathbf{2}$}, myleaf={\cmark\ \textsf{EMPTY}},
              edge label={node[midway,left,xshift=20pt,font=\scriptsize\sffamily]{F}}
            ],
            [{$\meven{[0,2]} x < 5 \mid \mathbf{2}$}, myleaf={\xmark\ \textsf{UNTIL}}],
          ]
        ]
      ]
    ]
  ]
\end{forest}
\caption{Basic tableau for $\even{[0,2]} x < 5$.}
\label{fig:F-example}
\end{wrapfigure}
Since $u_2$ contains a marked temporal operator, and we can apply no more expansion rules to it,
we call it a \emph{poised} node, and we apply the \textsf{STEP} rule to generate $u_3$:
$\globop$ is unmarked, and $x > 0$ is dropped, because we are about to construct a new instant of the signal.
The expansion rule, however, extracts it again, because $\phi$ requires $x > 0$ to hold at $t = 2$, too.

The last application of the \textsf{STEP} rule creates an empty node because the $\globop$ operator
has \emph{exhausted} its interval, i.e., the time of the current node is greater than its upper bound 2.
The \textsf{EMPTY} rule \emph{accepts} (\cmark) the resulting node $u_5$
because it contains no more temporal operators,
showing that we have found a signal that satisfies $\phi$.
Such a signal can be obtained by taking the tree branch from $u_0$ to $u_5$,
and finding for each time instant $t$ a variable assignment that satisfies real-valued constraints
in the poised node labeled with $t$ (i.e., the node to which we applied \textsf{STEP}).
Thus, a satisfying signal must satisfy $x>0$ for $t = 1, 2$.

Fig.~\ref{fig:F-example} shows the tableau for $\psi = \even{[0,2]} x < 5$.
At time $t = 0$, $\psi$ can be satisfied if, either, $x < 5$ at $t = 0$,
or $x < 5$ holds in a future instant.
We represent these two cases by applying the $\evenop$ expansion rule,
which adds two children $u_1$ and $u_2$ to $u_0$, both still with time $t(u_1) = t(u_2) = 0$:
$u_1$ is only labeled with $x < 5$, meaning that the requirement encoded by $\psi$ has been satisfied,
while $u_2$ is labeled with $\psi$, which is now \emph{marked} ($\marked{\evenop}$) because
it represents an obligation to be satisfied in the future (i.e., in one of $u_2$'s descendants).
$u_1$ has no more temporal operators and represents a signal that satisfies $\psi$:
it is a trivial signal made of one time instant, in which $x < 5$ holds.
The \textsf{EMPTY} rule accepts $u_1$.

We have proved that $\phi$ is satisfiable, and we could stop building the tableau.
However, we go on to illustrate how it works.
Since $u_2$ contains only marked temporal operators, we apply the \textsf{STEP} rule to proceed with the next time instant.
The new node $u_3$, only child of $u_2$, has $t(u_3) = 1$,
and contains temporal operators that were marked in $u_2$
which represent requirements that still have to be satisfied.
We keep building the tableau in the same way.
Applications of the expansion rule always create a left child that is accepted, 
and a right child that has further descendants, until we reach the rightmost node with $t = 2$.
This node contains a marked version of $\psi$, meaning that its requirement $x < 5$ should be satisfied at a time $t > 2$.
However, $\phi$ has \emph{exhausted} its interval: according to its semantics,
$x < 5$ must hold at a time $t \leq 2$.
Due to this contradiction, this node cannot lead to a signal that satisfies $\phi$ and
is \emph{rejected} (\xmark) by the \textsf{UNTIL} rule ($\evenop$ is a special case of $\untilop$).

\subsubsection{Formal definition}
\label{sec:basic-tableau-formal}
\color{black}
We assume w.l.o.g.\ that formula $\phi$ to be checked is in strict normal form.

\emph{Expansion} rules are the first to be applied to a new node,
and they create children in which one formula from the parent node is replaced with simpler formulas,
based on expansion laws for temporal operators.
Some of such rules may \emph{mark} an operator (e.g., $\marked{\suntilop}$),
signaling that its satisfaction has been postponed to a future time instant.
Table~\ref{tab:expansion-rules} lists all expansion rules.
In this section, we only consider rules in categories P and S.

When an expansion rule extracts nested temporal operators,
their intervals are updated according to their \emph{temporal expansion}.
The temporal expansion of a formula $\varphi$ is defined as
\begin{align*}
&\exp^t(\top) = \top \\
&\exp^t(f(R) \bowtie k) = f(R) \bowtie k && \text{with } \mathord{\bowtie} \in \set{>, =}\\
&\exp^t(\neg \varphi_1) = \neg \exp^t_I(\varphi_1) \\
&\exp^t(\varphi_1 \circ \varphi_2) = \exp^t(\varphi_1) \circ \exp^t(\varphi_2) && \text{with } \circ \in \set{\land, \lor} \\
&\exp^t(\anyunop{[a,b]}{\varphi}) = \anyunop{[a+t,b+t]}{\varphi} && \text{with } \mathsf{A} \in \set{\evenop, \globop} \\
&\exp^t(\anybinop{[a,b]}{\varphi_1}{\varphi_2}) = \anybinop{[a+t,b+t]}{\varphi_1}{\varphi_2} && \text{with } \mathsf{B} \in \set{\suntilop, \sreleaseop}
\end{align*}
\vspace{-1ex}

\begin{table}
\caption{Expansion rules.
A node $u$ with $\phi \in \Gamma(u)$ that satisfies the condition in the second column
is expanded into one or two children $u_1$ and $u_2$ such that
$t(u_1) = t(u_2) = t(u)$, $\Gamma(u_1) = (\Gamma(u) \setminus \{\phi\}) \cup \Gamma_\phi(u_1)$
and $\Gamma(u_2) = (\Gamma(u) \setminus \{\phi\}) \cup \Gamma_\phi(u_2)$.
$u_2$ is only created if $\Gamma_\phi(u_2) \neq \emptyset$.}
\label{tab:expansion-rules}
\centering
\begin{tabular}{c l l l l}
\toprule
Category & $\phi \in \Gamma(u)$ & cond. & $\Gamma_\phi(u_1)$ & $\Gamma_\phi(u_2)$ \\
\midrule
\multirow{2}{*}{P}
& $\varphi_1 \lor \varphi_2$ & & $\varphi_1$ & $\varphi_2$  \\
& $\varphi_1 \land \varphi_2$ & & $\varphi_1, \varphi_2$ & \\
\midrule
\multirow{4}{*}{S}
&  $\even{I} \varphi$ &  $t(u) \in I$ &  $\exp^{t(u)}(\varphi)$ &  $\meven{I} \varphi$ \\
& $\glob{I} \varphi$ & $t(u) \in I$ &  $\exp^{t(u)}(\varphi), \mglob{I} \varphi$ \\
& $\suntil{I}{\varphi_1}{\varphi_2}$ &  $t(u) \in I$ & $\exp^{t(u)}(\varphi_2)$ &  $\exp^{t(u)}(\varphi_1), \msuntil{I}{\varphi_1}{\varphi_2}$ \\
& $\srelease{I}{\varphi_1}{\varphi_2}$ &  $t(u) \in I$ &  $\exp^{t(u)}(\varphi_1 \land \varphi_2)$ &  $\exp^{t(u)}(\varphi_2), \msrelease{I}{\varphi_1}{\varphi_2}$ \\
\midrule
\multirow{5}{*}{J}
& $\even[J]{I} \varphi$ & $t(u) \in I$ &  $\exp^{t(u)}_I(\varphi)$ &  $\meven[J]{I} \varphi$ \\
&  $\glob[J]{I} \varphi$ & $t(u) \in I$ &  $\exp^{t(u)}_I(\varphi), \mglob[J]{I} \varphi$ \\
& $\suntil[J]{I}{\varphi_1}{\varphi_2}$ & $t(u) \in I$ & $\exp^{t(u)}_I(\varphi_2)$ & $\exp^{t(u)}_I(\varphi_1), \msuntil[J]{I}{\varphi_1}{\varphi_2}$ \\
& $\srelease[J]{I}{\varphi_1}{\varphi_2}$ & $t(u) \in I$ & $\exp^{t(u)}_I(\varphi_1 \land \varphi_2)$ & $\exp^{t(u)}_I(\varphi_2), \msrelease[J]{I}{\varphi_1}{\varphi_2}$ \\
\bottomrule
\end{tabular}
\end{table}

After repeated application of expansion rules, each branch ends with a node
that only contains terms and operators that either are marked,
or such that the lower bound of their interval is greater than the node's current time.
We call such nodes \emph{poised}.
We thus proceed by trying to build the next instant of the model signal, by applying the \textsf{STEP} rule.
\begin{description}
\item[\textsf{STEP}]
Let $u$ be a poised node.
If $\Gamma(u)$ contains any (marked or unmarked) temporal operators,
$u$ has one child $u'$ such that $t(u') = t(u) + 1$, and
\begin{align*}
\Gamma(u') =
\ & \set{\anyunop{I}{\varphi} \in \Gamma(u) \mid \mathsf{A} \in \set{\evenop, \globop}}
\cup \set{\anybinop{I}{\varphi_1}{\varphi_2} \in \Gamma(u) \mid \mathsf{B} \in \set{\suntilop, \sreleaseop}} \\
&\cup \set{\anyunop{[a,b]}{\varphi} \mid \stlunop{\marked{\mathsf{A}}}{[a,b]}{\varphi} \in \Gamma(u) \land \mathsf{A} \in \set{\evenop, \globop} \land t(u) < b} \\
&\cup
\set{\anybinop{[a,b]}{\varphi_1}{\varphi_2} \mid \stlbinop{\marked{\mathsf{B}}}{[a,b]}{\varphi_1}{\varphi_2} \in \Gamma(u) \land \mathsf{B} \in \set{\suntilop, \sreleaseop} \land t(u) < b}
\end{align*}
\end{description}

Let $u$ be a poised node, and $\overline{u}$ the branch from the root $u_0$ to $u$.
If $\overline{u}$ encodes a signal that satisfies $\phi$, $u$ is \emph{accepted};
if $\overline{u}$ cannot encode a prefix of any satisfying signal%
---i.e., there is no way to reach an accepted node by further developing the tree from $u$---%
then $u$ is \emph{rejected}.
The following \emph{termination} rules are applied to all poised nodes right before the \textsf{STEP} rule,
and they can accept or reject them.
\begin{description}
\item[\textsf{FALSE}] If $\neg \top \in \Gamma(u)$, node $u$ is rejected.

\item[\textsf{LOCALLY-UNSAT}]
If the set $\set{f(R) = k, f(R) > k \in \Gamma(u)}$ is inconsistent, $u$ is rejected.

\item[\textsf{UNTIL}]
If $\msuntil{[a,b]}{\varphi_1}{\varphi_2} \in \Gamma(u)$ or $\meven{[a,b]} \varphi \in \Gamma(u)$ with $b = t(u)$, then the node is rejected.

\item[\textsf{EMPTY}]
If $\Gamma(u)$ contains no temporal operators and no rejecting rules apply,
then $u$ is accepted.
\end{description}
The consistency of sets $\set{f(R) = k, f(R) > k \in \Gamma(u)}$ is checked by a solver for linear real arithmetic.

\subsection{The \textsf{JUMP} Rule}
\label{sec:jump}

\begin{wrapfigure}{r}
\tiny
\begin{forest}
  for tree={
    myleaf/.style={label=below:{\strut#1}}
  },
  [{$\glob{[0,10]} x > 5, \even{[0,11]} x < 0 \mid \mathbf{0}$}
    [{$\mglob{[0,10]} x > 5, \even{[0,11]} x < 0, x > 5 \mid \mathbf{0}$},
     edge label={node[midway,left,font=\scriptsize\sffamily]{G}}
      [{$\mglob{[0,10]} x > 5, x < 0, x > 5 \mid \mathbf{0}$},
       edge label={node[midway,right,xshift=20pt,font=\scriptsize\sffamily]{F}},
       myleaf={\xmark \ (\textsf{LOCALLY-UNSAT})}],
      [{$\mglob{[0,10]} x > 5, \meven{[0,11]} x < 0, x > 5 \mid \mathbf{0}$}
        [{$\glob{[0,10]} x > 5, \even{[0,11]} x < 0 \mid \mathbf{10}$},
         edge label={node[midway,left,font=\scriptsize\sffamily]{JUMP}}
          [{$\mglob{[0,10]} x > 5, \even{[0,11]} x < 0, x > 5 \mid \mathbf{10}$},
           edge label={node[midway,left,font=\scriptsize\sffamily]{G}}
            [{$\mglob{[0,10]} x > 5, x < 0, x > 5 \mid \mathbf{10}$},
             edge label={node[midway,right,xshift=20pt,font=\scriptsize\sffamily]{F}},
             myleaf={\xmark \ (\textsf{LOCALLY-UNSAT})}],
            [{$\mglob{[0,10]} x > 5, \meven{[0,11]} x < 0, x > 5 \mid \mathbf{10}$}
              [{$\even{[0,11]} x < 0 \mid \mathbf{11}$},
               edge label={node[midway,right,font=\scriptsize\sffamily]{STEP}},
                [{$x < 0 \mid \mathbf{11}$},
                 edge label={node[midway,right,xshift=10pt,font=\scriptsize\sffamily]{F}},
                 myleaf={\cmark \ (\textsf{EMPTY})}],
                [{$\meven{[0,11]} x < 0 \mid \mathbf{11}$}, myleaf={\xmark \ (\textsf{UNTIL})}]
              ]
            ]
          ]
        ]
      ]
    ]
  ]
\end{forest}
\caption{Tableau for formula $\glob{[0,10]} x > 5 \land \even{[0,11]} x < 0$.}
\label{fig:tableau-example}
\end{wrapfigure}
Consider, again, the example of Fig.~\ref{fig:G-example}.
The labels of nodes $u_1$ and $u_2$ are repeated identically in $u_2$'s descendants,
except for their time instant: these are \emph{repeated} nodes.
The same can be said of nodes $u_0$, $u_1$, and $u_2$ in Fig.~\ref{fig:F-example}.
For formulas containing temporal operators with long intervals,
the basic tableau contains many repeated nodes, leading to unnecessary redundancy.
For instance, the basic tableau for $\phi' = \glob{[0,20]}{(x > 0)}$
consists of the repetition of nodes $u_1$ and $u_2$ for 20 times.
This redundancy is unnecessary: once we have created $u_1$ and $u_2$,
we already know what the rest of the tableau will look like.
To avoid redundancy and speed up the construction of the tableau,
we introduce the \textsf{JUMP} rule, which we apply instead of the \textsf{STEP} rule
to ``skip'' nodes with repeated labels, making a temporal jump by more than one time unit.

\subsubsection{Example}

Fig.~\ref{fig:tableau-example} shows the
tableau for formula $\glob{[0,10]} x > 5 \land \even{[0,11]} x < 0$.
First, we apply to the root the expansion rule for the $\globop$ operator,
which extracts $x > 5$ and marks the $\globop$ operator.
Then, the rule for $\evenop$ creates two children:
one (left) in which $\even{[0,11]} x < 0$ is satisfied, as $x < 0$ holds;
and one (right) in which its satisfaction is postponed, and the $\evenop$ operator is left marked.
The left child is immediately rejected by the \textsf{LOCALLY-UNSAT} rule,
because the constraints $x < 0$ and $x > 5$ are contradictory.
The right child is a poised node, and we can apply the \textsf{STEP} rule to it.

\begin{wrapfigure}{r}
\begin{adjustbox}{max width=\textwidth}
\tiny
\begin{forest}
  for tree={
    myleaf/.style={label=below:{\strut#1}}
  },
  [{$\glob[J_0]{[0,10]} \even{[0,1]} p \mid \mathbf{0}$}
    [{$\mglob[J_0]{[0,10]} \even{[0,1]} p, \even[J_1]{[0,1]} p \mid \mathbf{0}$},
       edge label={node[midway,left,font=\scriptsize\sffamily]{G}},
      [{$\mglob[J_0]{[0,10]} \even{[0,1]} p, p \mid \mathbf{0}$}, myleaf={\dots},
        edge label={node[midway,right,xshift=20pt,font=\scriptsize\sffamily]{F}},
      ],
      [{$\mglob[J_0]{[0,10]} \even{[0,1]} p, \meven[J_1]{[0,1]} p \mid \mathbf{0}$}
        [{$\glob[J_0]{[0,10]} \even{[0,1]} p, \even[J_1]{[0,1]} p \mid \mathbf{1}$},
         edge label={node[midway,left,font=\scriptsize\sffamily]{STEP}}
          [{$\mglob[J_0]{[0,10]} \even{[0,1]} p, \even[J_1]{[0,1]} p, \even[J_1]{[1,2]} p \mid \mathbf{1}$},
            edge label={node[midway,left,font=\scriptsize\sffamily]{G}},
            [{$\mglob[J_0]{[0,10]} \even{[0,1]} p, p, \even[J_1]{[1,2]} p \mid \mathbf{1}$},
              edge label={node[midway,right,xshift=20pt,font=\scriptsize\sffamily]{F}},
              [{$\mglob[J_0]{[0,10]} \even{[0,1]} p, p \mid \mathbf{1}$}, myleaf={\dots},
                edge label={node[midway,right,xshift=20pt,font=\scriptsize\sffamily]{F}}
              ],
              [{$u_1: \mglob[J_0]{[0,10]} \even{[0,1]} p, p, \meven[J_1]{[1,2]} p \mid \mathbf{1}$}
                [{\dots}
                  [{$u_{10}: \mglob[J_0]{[0,10]} \even{[0,1]} p, p, \meven[J_1]{[10,11]} p \mid \mathbf{10}$}
                    [{$\even[J_1]{[10,11]} p \mid \mathbf{11}$},
                      edge label={node[midway,left,font=\scriptsize\sffamily]{STEP}},
                      [{$p \mid \mathbf{11}$}, myleaf={\cmark \ (\textsf{EMPTY})},
                        edge label={node[midway,right,xshift=5pt,font=\scriptsize\sffamily]{F}}
                      ]
                      [{$\meven[J_1]{[10,11]} p \mid \mathbf{11}$}, myleaf={\xmark \ (\textsf{UNTIL})}]
                    ]
                  ]
                ]
              ]
            ],
            [{$\mglob[J_0]{[0,10]} \even{[0,1]} p, \meven[J_1]{[0,1]} p, \even[J_1]{[1,2]} p \mid \mathbf{1}$}, myleaf={\xmark \ (\textsf{UNTIL})}]
          ]
        ]
      ]
    ]
  ]
\end{forest}
\end{adjustbox}
\caption{$p = x > 0$, $J_0 = [-1,-1]$, and $J_1 = [0, 10]$.}
\label{fig:GF-tree}
\end{wrapfigure}
If we proceeded as in Sec.~\ref{sec:basic-tableau},
we would obtain a tree consisting of nine more repetitions of the first three levels of Fig.~\ref{fig:tableau-example}.
Instead, we use a new rule, that ``jumps'' directly to time 10, when the $\globop$ operator reaches the end of its interval.
This rule, called the \textsf{JUMP} rule, picks the smallest time bound $k$ of any operator that is greater than the current node's time,
and creates a child node that has the same label as the node obtained by applying the \textsf{STEP} rule repeatedly until reaching $k$.

We apply again the expansion rules, and the right child created by the $\evenop$ rule is a poised node to which the \textsf{STEP} rule applies.
Since the $\globop$ operator has exhausted its interval, the \textsf{STEP} rule creates a node that contains only $\even{[0,11]} x < 0$.
One application of the $\evenop$ expansion rule creates an accepted node,
showing that $\even{[0,11]} x < 0$ can be satisfied by imposing $x < 0$ at time 11.
We can obtain a signal that satisfies the root formula by visiting the tree from the root to the accepted node, and collecting real constraints in poised nodes.
The signal consists of solutions to these constraint sets.
To reconstruct signal instants skipped by the \textsf{JUMP} rule,
we just repeat the solution to constraints in the node to which we applied it
($x > 5$ in the example).

Fig.~\ref{fig:GF-tree} shows part of the tableau for formula $\glob[J_0]{[0,10]} \even{[0,1]} x > 0$.
Nodes are created with the same rules as in the previous examples, so we just highlight a few aspects.
Nodes connected to their parents by an edge labeled with \textsf{G} derive from the extraction of the \textsf{F} operator from the outer \textsf{G}.
Their (subscript) interval bounds are increased by the current time---%
for instance, the second application of the \textsf{G} expansion rule extracts $\even[J_1]{[1,2]} x > 0$.
Moreover, temporal operators are superscripted with the time interval of the operator from which they were extracted (or $[-1,-1]$ if they have no parent).
Let us ignore this superscript for now (we will need it later on).

The $\globop$ operator produces a new instance of $\even[J_1]{[t,t+1]} x > 0$ at every time instant.
These $\evenop$ operators produce two children each.
Thus, the tree grows exponentially.

However, $\evenop$ operators keep disappearing, because they must be satisfied within their intervals.
Thus, the number of occurrences of $\even[J_1]{[t,t+1]} x > 0$ in nodes is bounded by their intervals' width (2 in this case).
Moreover, nodes become highly redundant, even though not exactly identical:
for instance, nodes $u_1$ and $u_{10}$ only differ in the interval of the $\evenop$ operator.
Thus, we cannot just jump from $u_1$ to a node differing only in the time counter,
but the \textsf{JUMP} rule must update the intervals of $\evenop$ operators extracted from $\glob[J_0]{[0,10]} \even{[0,1]} x > 0$ accordingly.

\subsubsection{Formal definition}
We now describe the tableau with the \textsc{JUMP} rule.
In the following, we omit unary operators $\globop$ and $\evenop$ to simplify notation.
Rules for them can be easily obtained through the usual equivalences
$\glob{I} \varphi \equiv \srelease{I}{\neg \top}{\varphi}$ and $\even{I} \varphi \equiv \suntil{I}{\top}{\varphi}$.

To use the new rule, the tableau decorates (in the superscript) nested temporal operators
with the interval of the operator in which they are nested,
which we call the \emph{parent} formula.
This interval is used to keep track of whether the current time falls within
the interval in which the parent of a temporal operator is active.
For instance, $\even[J_1]{[t,t+1]} x > 0$ is decorated with $J_1$,
that is the interval in which the parent formula $\glob[J_0]{[0,10]} \even{[0,1]} x > 0$ is active.
Temporal operators in $\pcl(\Gamma(u_0))$%
---the root label---are decorated with the interval $[-1,-1]$,
meaning that they have no parent.

We use the expansion rules in categories P and J in Table~\ref{tab:expansion-rules},
and re-define the temporal expansion of a formula $\varphi$ with respect to interval $I$ as follows (it remains the same for propositional operators):
\begin{align*}
&\exp^t_I(\anybinop{[a,b]}{\varphi_1}{\varphi_2}) = \anybinop[I]{[a+t,b+t]}{\varphi_1}{\varphi_2} && \text{with } \mathsf{B} \in \set{\suntilop, \sreleaseop}
\end{align*}

We apply the \textsf{STEP} rule to a poised node $u$
only if $\Gamma(u)$ contains at least one operator
${\msuntil[J]{[a,b]}{\varphi_1}{\varphi_2}}$ or $\msrelease[J]{[a,b]}{\varphi_2}{\varphi_1}$
such that $t(u) \not\in J$ and:
\begin{itemize}
\item $t(u) = b$, or
\item there exists a formula $\psi \in \pcl(\varphi_1)$ such that $t(u) < a + I_u(\psi)$.
\end{itemize}
Otherwise, the \textsf{JUMP} rule is applied:
\begin{description}
\item[\textsf{JUMP}]
If all marked operators
$\msuntil[J]{[a,b]}{\varphi_1}{\varphi_2}$ or $\msrelease[J]{[a,b]}{\varphi_2}{\varphi_1}$ in $\Gamma(u)$
with $t(u) \not\in J$ are such that
\begin{itemize}
\item $t(u) < b$, and
\item for all formulas $\psi \in \pcl(\varphi_1)$ we have $t(u) \geq a + I_u(\psi)$,
\end{itemize}
then $u$ has one child $u'$ such that $t(u') = \min \set{t \in K(u) \mid t > t(u)}$,
where
\[
K(u) = \set{a, b \mid \anybinop[J]{[a,b]}{\varphi_1}{\varphi_2} \in \Gamma(u) \land \mathsf{B} \in \set{\suntilop, \sreleaseop, \marked{\suntilop}, \marked{\sreleaseop}} \land t(u) \not\in J},
\]
and, setting $k = t(u') - t(u)$, we have
\begin{align}
\Gamma(u') =
&\ \set{\anybinop[J]{I}{\varphi_1}{\varphi_2} \in \Gamma(u) \mid \mathsf{B} \in \set{\suntilop, \sreleaseop} \land t(u) \not\in J} \label{eq:jump-outer-unmarked} \\
&\cup
\set{\anybinop[J]{[a+k, b+k]}{\varphi_1}{\varphi_2} \mid \anybinop[J]{[a,b]}{\varphi_1}{\varphi_2} \in \Gamma(u) \land \mathsf{B} \in \set{\suntilop, \sreleaseop} \land t(u) \in J} \label{eq:jump-derived-unmarked} \\
&\cup
\set{\anybinop[J]{[a,b]}{\varphi_1}{\varphi_2} \mid \stlbinop[J]{\marked{\mathsf{B}}}{[a,b]}{\varphi_1}{\varphi_2} \in \Gamma(u) \land \mathsf{B} \in \set{\suntilop, \sreleaseop} \land t(u) \not\in J \land t(u) < b} \label{eq:jump-outer-marked} \\
&\cup
\set{\anybinop[J]{[a+k,b+k]}{\varphi_1}{\varphi_2} \mid \stlbinop[J]{\marked{\mathsf{B}}}{[a,b]}{\varphi_1}{\varphi_2} \in \Gamma(u) \land \mathsf{B} \in \set{\suntilop, \sreleaseop} \land t(u) \in J \land t(u) < b} \label{eq:jump-derived-marked}
\end{align}
\end{description}
Note that the conditions for applying the \textsf{STEP} and \textsf{JUMP} rules are mutually exclusive,
so one of them can always be applied to a poised node.

The \textsf{JUMP} rule can only be applied if $t(u) < b$ for all marked operators,
so that consistency is checked with the basic tableau at all interval upper bounds.
Moreover, the jump length $k$ is computed so that the child node's time is the closest
bound of any operator that is greater than the current time.
These two constraints make sure that we always check consistency of all possible combinations of operators that can be active at any time.

Sets \eqref{eq:jump-outer-unmarked} and \eqref{eq:jump-outer-marked}
that compose the label of the new child $u'$ implement straightforward jumps,
like the ones in Fig.~\ref{fig:tableau-example}.
Condition $t(u) \not\in J$ selects only operators that are either not nested ($J = [-1,-1]$),
or that have been extracted from the arguments of other temporal operators
that are no longer active, because $u$'s time falls after their interval.
Such operators are placed in $\Gamma(u')$ without changing their time intervals,
only unmarking them if they are marked.

Sets \eqref{eq:jump-derived-unmarked} and \eqref{eq:jump-derived-marked} cover
cases like the one in Fig.~\ref{fig:GF-tree}.
Here, to jump between $u_1$ and $u_{10}$ we cannot just put $\meven[J_1]{[1,2]} p$
into $\Gamma(u_{10})$, but we need to increment \emph{both bounds} of its interval by the jump length $k = 9$.
This kind of more complex jump, however, can only be performed once the tableau
has been explored enough by the \textsf{STEP} rule to have tried all possible combinations
of choices of whether to satisfy $\untilop$ and $\releaseop$ operators,
i.e., all possible applications of their expansion rules.
The second condition for applying the \textsf{JUMP} rule,
requiring that $t(u)$ has reached or exceeded the end of the interval of all derived operators,
ensures this constraint.
In the example, the \textsf{JUMP} rule can be applied only at time 1,
when $\even{[0,1]} p$ has exhausted its interval.

Note that $\Gamma(u_{10})$, as obtained by applying the \textsf{JUMP} rule,
is not exactly the same as $u_{10}$'s label in the figure:
real-numbered constraints are dropped, and marked operators are unmarked.
Indeed, consistency has already been checked in the node $u$
to which we applied the \textsf{JUMP} rule,
so it needs not be checked again.
Unmarked operators are marked again by expansion rules.

Termination rules are the same as in the basic tableau.


\subsubsection{Reconstructing signals}
\label{sec:reconstructing-signals}
If the tableau contains an accepted node $u_n$, a satisfying signal
can be obtained from the branch $\mathbf{u} = u_0 \dots u_n$ that goes from the root $u_0$ to $u_n$.
Let $\textbf{p} = p_0 \dots p_m$ be the sequence of indices
such that $u_{p_0} \dots u_{p_m}$ is the sequence of poised nodes in $\mathbf{u}$.
For each $0 \leq i \leq m$, we ask the linear-real-arithmetic solver
for a solution to the constraint set $\set{f(R) = k, f(R) > k \in \Gamma(u_{p_i})}$.
The resulting variable assignment will constitute time $t = p_i$ of our satisfying signal.
If the \textsf{JUMP} rule was applied to some node $u_{p_j}$, for $0 \leq j < m$, then we have $p_{j+1} > p_j$.
We associate to time instants $t$, with $p_j < t < p_{j+1}$,
of the signal the variable assignment obtained for $u_{p_j}$ as described before.
For instance, from Fig.~\ref{fig:tableau-example} we obtain a signal
that satisfies $x > 5$ at times $t = 0, \dots, 10$, and $x < 0$ at $t = 11$,
such as $x = 6$ for $t = 0, \dots, 10$ and $x = -1$ at $t = 11$.

\color{black}

\subsection{Soundness}

We prove that the tableau is sound, meaning that if a complete tableau has an accepted branch,
then the formula is satisfiable, and the set of requirements it represents is consistent.
To simplify the proofs, we only consider $\suntilop$ and $\sreleaseop$ operators,
from which $\globop$ and $\evenop$ are derived.
We start by proving termination.

\begin{definition}
\label{def:time-horizon}
The \emph{Time Horizon} $\horiz(\phi)$ of a formula $\phi$ is defined as follows:
\begin{align*}
&\horiz(\top) = \horiz(f(R) > k) = 0 &
&\horiz(\neg \varphi_1) = \horiz(\varphi_1) \\
&\horiz(\varphi_1 \circ \varphi_2) = \max(\horiz(\varphi_1), \horiz(\varphi_2)) && \text{with } \circ \in \set{\land, \lor} \\
&\horiz(\anybinop[J]{[a,b]}{\varphi_1}{\varphi_2}) = b + \max(\horiz(\varphi_1), \horiz(\varphi_2)) && \text{with } \mathsf{B} \in \set{\suntilop, \sreleaseop}
\end{align*}
\end{definition}

\begin{theorem}[Termination]
The complete tableau for a formula $\phi$ is a finite tree.
\end{theorem}
\begin{proof}
Expansion rules, and step and jump rules only add a finite number of children (at most two) to each node.
Thus, the tableau can be infinite only if it has at least an infinitely long branch.
First, note that expansion rules may only add a finite number of nodes to a branch,
because they only add children containing simpler or marked formulas.
Further, we prove that the step and jump rules can be applied only a finite number of times in each branch.
To do so, we prove that for any node $u$ and for all temporal operators $\psi \in \Gamma(u)$,
we have $I_u(\psi) \leq \horiz(\phi)$.

New temporal operators can be introduced into nodes by the expansion rules.
For any formula $\psi \in \Gamma(u)$, for all temporal operators
$\varphi \in \Gamma_\psi(u_1) \cup \Gamma_\psi(u_2)$,
we have $I_u(\varphi) \leq \horiz(\psi)$.
This fact is only non-trivial for temporal operators.
For instance, let $\psi = \suntil[J]{I}{\varphi_1}{\varphi_2}$,
and let $\anybinop{[a,b]}{\varphi_3}{\varphi_4} \in \pcl(\varphi_1) \cup \pcl(\varphi_2)$,
with $\mathsf{B} \in \set{\suntilop, \sreleaseop}$.
In this case, we have $\anybinop[I]{[a+t,b+t]}{\varphi_3}{\varphi_4} \in \pcl(\varphi)$
for some $\varphi \in \Gamma_\psi(u_1) \cup \Gamma_\psi(u_2)$.
The expansion rule is triggered only if $t(u) \in I$, hence $t \leq I_u(\psi)$.
It is easy to prove by induction that, since $\anybinop{[a,b]}{\varphi_3}{\varphi_4}$ is a subformula of $\psi$,
we have $I_u(\psi) + b \leq I_u(\psi) + \max(\horiz(\varphi_1), \horiz(\varphi_2)) \leq \horiz(\psi) $,
hence $b+t \leq I_u(\psi) + b \leq \horiz(\psi)$.
The proof is analogous for the release operator.

This fact, together with $\horiz(\psi') \leq \horiz(\psi)$ for any subformula $\psi'$ of $\psi$,
concludes the proof.

The step rule leaves bounds of temporal operators unchanged.
The jump rule may increase the upper bound of an operator
$\anybinop[J]{[a,b]}{\varphi_1}{\varphi_2}$ with $\mathsf{B} \in \set{\suntilop, \sreleaseop}$
and $J = [c,d]$ by an amount $k$ such that $k \leq d$, because $k \in K(u)$,
and if $\anybinop[J]{[a,b]}{\varphi_1}{\varphi_2} \in \Gamma(u)$,
it must be derived by the expansion rules by an operator $\varphi \in \Gamma(u)$
such that $I(\varphi) = [c,d]$.
Thus, $\anybinop[J]{[a,b]}{\varphi_1}{\varphi_2}$ is a subformula of $\varphi$,
and $b + d \leq \horiz(\varphi)$, hence $b + k \leq \horiz(\varphi)$.

For any node $u$, the node $u'$ added by a step or a jump rule is such that $t(u') > t(u)$,
and for all $\varphi \in \Gamma(u')$ we have either $I_l(\varphi) > t(u)$ for non-marked operators,
or $I_u(\varphi) > t(u)$.
Since $I_u(\varphi) < \horiz(\phi)$, at some point a node $u_f$ is reached such that
either $\Gamma(u_f)$ contains no temporal operators, or $t(u_f) \geq \horiz(\phi)$.
In both cases, $u_f$ is a leaf.
\end{proof}

To prove the soundness of the tableau, we introduce the concept of \emph{model} of a formula $\phi$, which formalizes the proof obligations that need to hold in each time instant of a satisfying signal.
In fact, in Lemma~\ref{lemma:model} we prove that we can build a satisfying signal from a model for $\phi$, and \emph{vice versa}.
In Thm.~\ref{thm:soundness} we exploit this fact by showing that if the tableau for $\phi$ has an accepted branch,
then it has a model, and therefore a satisfying signal.
\begin{definition}
\label{def:model}
A \emph{model} for a formula $\phi$ is a sequence of \emph{atoms}
$\mathbf{\Delta} = \Delta_0 \Delta_1 \dots \Delta_n$,
with $n \leq \horiz(\phi)$, that are the smallest sets of formulas that satisfy the following rules:
$\phi \in \Delta_0$;\\
for each $i = 0, \dots, n$, we have:
\begin{enumerate}
\item \label{item:model-or}
    if $\varphi_1 \lor \varphi_2 \in \Delta_i$ then $\varphi_1 \in \Delta_i$ or $\varphi_2 \in \Delta_i$
\item \label{item:model-and}
    if $\varphi_1 \land \varphi_2 \in \Delta_i$ then $\varphi_1 \in \Delta_i$ and $\varphi_2 \in \Delta_i$
\end{enumerate}
and for each $j = 0, \dots, n-1$ we have
\begin{enumerate}[resume]
\item \label{item:model-temp-lta}
    if $\anybinop{[a,b]}{\varphi_1}{\varphi_2} \in \Delta_j$,
    with $\mathsf{B} \in \set{\suntilop, \sreleaseop}$, and $j < a$ then
    $\anybinop{[a,b]}{\varphi_1}{\varphi_2} \in \Delta_{j+1}$,
\item \label{item:model-suntil-atb}
    if $\suntil{[a,b]}{\varphi_1}{\varphi_2} \in \Delta_j$ and $a \leq j < b$ then either
    $\exp^j(\varphi_2) \in \Delta_j$ or
    ($\exp^j(\varphi_1) \in \Delta_j$ and $\suntil{[a,b]}{\varphi_1}{\varphi_2} \in \Delta_{j+1}$),
\item \label{item:model-suntil-eqb}
    if $\suntil{[a,b]}{\varphi_1}{\varphi_2} \in \Delta_j$ and $j = b$ then
    $\exp^j(\varphi_2) \in \Delta_j$,
\item \label{item:model-srelease-atb}
    if $\srelease{[a,b]}{\varphi_1}{\varphi_2} \in \Delta_j$ and $a \leq j < b$ then either
    $\exp^j(\varphi_1 \land \varphi_2) \in \Delta_j$ or
    ($\exp^j(\varphi_2) \in \Delta_j$ and $\srelease{[a,b]}{\varphi_1}{\varphi_2} \in \Delta_{j+1}$),
\item \label{item:model-srelease-eqb}
    if $\suntil{[a,b]}{\varphi_1}{\varphi_2} \in \Delta_j$ and $j = b$ then
    $\exp^j(\varphi_1 \land \varphi_2) \in \Delta_j$,
\end{enumerate}
where $\exp^j$ is defined as in Sec.~\ref{sec:basic-tableau-formal}.
\end{definition}

\begin{lemma}
\label{lemma:model}
A formula $\phi$ is satisfiable iff it has a model $\mathbf{\Delta} = \Delta_0 \Delta_1 \dots \Delta_n$
such that for all $i = 1, \dots, n$ the set $\set{f(R) = k, f(R) > k \in \Delta_i}$ is consistent.
\end{lemma}
\begin{proof}
Given a model for $\phi$, we can build a signal $w$ that satisfies $\phi$ such that $w(i)$ is an assignment that satisfies
$\set{f(R) = k, f(R) > k \in \Delta_i}$ for each $i = 1, \dots, n$.

The fact that such a signal satisfies $\phi$ can be proved by induction on the syntactic structure of $\phi$
by applying the following laws:
\begin{align*}
&(w, t) \models \anybinop{[a,b]}{\varphi_1}{\varphi_2}
&&\text{iff}
&&(w, t+1) \models \anybinop{[a,b]}{\varphi_1}{\varphi_2}
\text{ for $\mathsf{B} \in \set{\suntilop, \sreleaseop}$}
&\text{if $t < a$} \\
&(w, t) \models \suntil{[a,b]}{\varphi_1}{\varphi_2}
&&\text{iff}
&&\begin{aligned}[t]
&(w, t) \models \exp^t(\varphi_2) \\
&\lor (w, t) \models \exp^t(\varphi_1) \land (w, t+1) \models \suntil{[a,b]}{\varphi_1}{\varphi_2}
\end{aligned}
&\text{if $a \leq t < b$} 
\end{align*}
\begin{align*}
&(w, t) \models \suntil{[a,b]}{\varphi_1}{\varphi_2}
&&\text{iff}
&&(w, t) \models \exp^t(\varphi_2)
&\text{if $t = b$} \\
&(w, t) \models \srelease{[a,b]}{\varphi_1}{\varphi_2}
&&\text{iff}
&&\begin{aligned}[t]
&(w, t) \models \exp^t(\varphi_1) \land (w, t) \models \exp^t(\varphi_2) \\
&\lor (w, t) \models \exp^t(\varphi_2) \land (w, t+1) \models \srelease{[a,b]}{\varphi_1}{\varphi_2}
\end{aligned}
&\text{if $a \leq t < b$} \\
&(w, t) \models \srelease{[a,b]}{\varphi_1}{\varphi_2}
&&\text{iff}
&&(w, t) \models \exp^t(\varphi_1) \land (w, t) \models \exp^t(\varphi_2)
&\text{if $t = b$}
\end{align*}
These laws can be proved similarly to the well-known expansion laws for the LTL until and release operators~\cite{BaierK08},
but adapting them for STL's timed operators.

If a signal satisfies $\phi$, one can build a model for it by applying the laws above bottom-up 
w.r.t.\ the syntactic structure of the formula.
\end{proof}

Before proceeding with the soundness proof, we highlight the fact that nested temporal operators cause periodicity in sufficiently deep tableau branches.
We illustrate this phenomenon in Fig.~\ref{fig:GF-tree},
and we formalize it here to justify the correctness of the \textsf{JUMP} rule.
\begin{lemma}
\label{lemma:operator-ladder}
Let $\suntil[J]{I}{\varphi_1}{\varphi_2}$ or $\srelease[J]{I}{\varphi_2}{\varphi_1}$ in $\Gamma(u)$
with $I = [a,b]$ for some poised node $u$ in the basic tableau such that $t(u) = a$.

If the conjunction of the formulas in $\Gamma(u)$ is satisfiable, there exists a branch in which
for all formulas $\anybinop{[c,d]}{\varphi_1}{\varphi_2} \in \pcl(\varphi_1)$,
for all poised nodes $u'$ such that $a + d \leq t(u') \leq b + d$, we have either
$\anybinop[I]{[t(u')+c-i,t(u')+d-i]}{\varphi_1}{\varphi_2} \in \Gamma(u')$ or
$\stlbinop[I]{\marked{\mathsf{B}}}{[t(u')+c-i,t(u')+d-i]}{\varphi_1}{\varphi_2} \in \Gamma(u')$
for all $i = 0, 1, \dots, d$.
\end{lemma}
\begin{proof}
The proof can be carried out by induction, observing that if we unfold a branch starting from $u$
without ever using the \textsf{JUMP} rule,
and by always choosing the $\Gamma_\phi(u_2)$ part of expansion rules when possible
(i.e., when the current time is strictly lower than an operator's upper bound),
new instances of nested temporal operators are extracted at each step,
with their time bounds increased by 1.
Since operators whose upper bound is lower than a node's time disappear,
we reach a point (namely after time $a + d$) in which extracted operators become stationary.
I.e., we have $\anybinop[I]{[t(u')+c-i,t(u')+d-i]}{\varphi_1}{\varphi_2}$
(or the marked counterpart) for $i = 0, 1, \dots, d$ and, at each further step,
$\anybinop[I]{[t(u')+c-d,t(u')]}{\varphi_1}{\varphi_2}$ disappears (with $i = d$),
while a new $\anybinop[I]{[t(u')+c,t(u')+d]}{\varphi_1}{\varphi_2}$ appears ($i = 0$).
\end{proof}

Given a set of formulas $\Gamma$ from tableau node labels,
we define $\unmark(\Gamma)$ as the set obtained by unmarking marked operators,
and removing interval superscripts from temporal operators in $\Gamma$.
For instance, if $\suntil[J]{I}{\varphi_1}{\varphi_2}$ or $\msuntil[J]{I}{\varphi_1}{\varphi_2}$ are in $\Gamma$,
they are both replaced by $\suntil{I}{\varphi_1}{\varphi_2}$.
Instances of the release operator are replaced similarly, and all other operators remain unchanged.

\begin{theorem}[Soundness]
\label{thm:soundness}
If the tableau for formula $\phi$ has an accepted branch, $\phi$ is satisfiable.
\end{theorem}
\begin{proof}
We show how to build a model for $\phi$ from an accepted branch of the tableau,
as we did in Sec.~\ref{sec:reconstructing-signals}, but more formally.
Let $\mathbf{u} = u_0 \dots u_n$ be the accepted branch,
and let $\textbf{p} = p_0 \dots p_m$ be the sequence of indices
such that $u_{p_0} \dots u_{p_m}$ is the sequence of poised nodes in $\mathbf{u}$.

For each $i = 1, \dots, m$, we define $\Delta(u_{p_i})$
as the set $\bigcup_{j = p_{i-1} + 1}^{p_i} \unmark(\Gamma(u_j))$ 
and $\Delta(u_{p_0})$ as $\bigcup_{j = 0}^{p_0} \unmark(\Gamma(u_j))$.
Note that all $\Delta(u_{p_i})$'s satisfy rules \ref{item:model-or} and \ref{item:model-and}
from Def.~\ref{def:model}.

We build a model $\mathbf{\Delta} = \Delta_0 \dots \Delta_{t_f}$ for $\phi$ inductively.
We set $\Delta_0 = \Delta(u_{p_0})$.
Note that $\phi \in \Delta_0$ and $t(u_{p_0}) = 0$.
Then, let $\Delta_t = \Delta(u_{p_i})$, with $t(u_{p_i}) = t$.

If the \textsf{STEP} rule can be applied to $u_{p_i}$, then we set $\Delta_{t+1} = \Delta(u_{p_{i+1}})$.
Let $\suntil{[a,b]}{\varphi_1}{\varphi_2} \in \Delta_t$.
If $t < a$, then $t(u_{p_i}) \not\in [a,b]$, and expansion rules from Table~\ref{tab:expansion-rules}
do not apply to $\suntil[J]{[a,b]}{\varphi_1}{\varphi_2}$, leaving it unmarked.
Thus, the \textsf{STEP} rule just propagates it to $\Gamma(u_{p_i+1})$,
so $\suntil{[a,b]}{\varphi_1}{\varphi_2} \in \Delta_{t+1}$.
Hence, rule \ref{item:model-temp-lta} of Def.~\ref{def:model} is satisfied.
If $a \leq t < b$, then expansion rules apply, and we have either
(a) $\exp^t(\varphi_2) \in \Delta_t$,
or (b) $\exp^t(\varphi_1) \in \Delta_t$ and $\msuntil[J]{I}{\varphi_1}{\varphi_2} \in \Gamma(u_{p_i})$.
In case (a), rule \ref{item:model-suntil-atb} is clearly satisfied.
In case (b), the \textsf{STEP} rule inserts $\suntil[J]{I}{\varphi_1}{\varphi_2} \in \Gamma(u_{p_i + 1})$,
hence $\suntil{I}{\varphi_1}{\varphi_2} \in \Delta(u_{p_{i+1}})$,
which satisfies rule \ref{item:model-suntil-atb}.
A similar argument shows that $\Delta_{t+1}$ satisfies the rules of Def.~\ref{def:model}
when $t = b$ and for the $\releaseop$ operator.

If the \textsf{JUMP} rule is applied to $u_{p_i}$, then $t(u_{p_{i+1}}) > t + 1$.
We set $\Delta_{t(u_{p_{i+1}})} = \Delta(u_{p_{i+1}})$
and synthesize the atoms from $\Delta_{t+1}$ to $\Delta_{t(u_{p_{i+1}})-1}$.
For each $j = 1, \dots, k$, where $k = t(u_{p_{i+1}}) - t$, we define
\begin{align*}
\delta_{j+1} =
&\set{\anybinop[J]{I}{\varphi_1}{\varphi_2} \in \Gamma(u_{p_i}) \mid \mathsf{B} \in \set{\suntilop, \sreleaseop} \land t(u_{p_i}) \not\in J} \\
&\cup
\set{\anybinop[J]{[a+j, b+j]}{\varphi_1}{\varphi_2} \mid \anybinop[J]{[a,b]}{\varphi_1}{\varphi_2} \in \Gamma(u_{p_i}) \land \mathsf{B} \in \set{\suntilop, \sreleaseop} \land t(u_{p_i}) \in J} \\
&\cup
\set{\anybinop[J]{[a,b]}{\varphi_1}{\varphi_2} \mid \stlbinop[J]{\marked{\mathsf{B}}}{[a,b]}{\varphi_1}{\varphi_2} \in \Gamma(u_{p_i}) \land \mathsf{B} \in \set{\suntilop, \sreleaseop} \land t(u_{p_i}) \not\in J \land t(u_{p_i}) < b} \\
&\cup
\set{\anybinop[J]{[a+j,b+j]}{\varphi_1}{\varphi_2} \mid \stlbinop[J]{\marked{\mathsf{B}}}{[a,b]}{\varphi_1}{\varphi_2} \in \Gamma(u_{p_i}) \land \mathsf{B} \in \set{\suntilop, \sreleaseop} \land t(u_{p_i}) \in J \land t(u_{p_i}) < b}
\end{align*}
We then build $\overline{\Delta}_{t+j}$ by applying expansion rules from Table~\ref{tab:expansion-rules} to $\delta_j$ until saturation.
When applying the rules for $\suntil[J]{[a,b]}{\varphi_1}{\varphi_2}$
and $\srelease[J]{[a,b]}{\varphi_1}{\varphi_2}$, we always choose $\Gamma_\phi(u_2)$ unless $t+j = b$;
when applying the rule for $\varphi_1 \lor \varphi_2$,
we extract $\varphi_1$ or $\varphi_2$ by repeating the choice that was made during the expansions
that lead to $u_{p_i}$
(note that formulas of the form $\varphi_1 \lor \varphi_2$ can only appear after being extracted from a temporal operator, and all temporal operators in $\Gamma(u_{p_{i+1}})$ must be already present in $\Gamma(u_{p_i})$ because of how the jump length $k$ is defined).
Finally, we set $\Delta_{t+j} = \unmark(\overline{\Delta}_{t+j})$.

We show that $\Delta_{t+j-1}$ satisfies the rules of Def.~\ref{def:model} for all $j = 1, \dots, k$.
Note that for any formula $\anybinop[J]{[a,b]}{\varphi_1}{\varphi_2} \in \overline{\Delta}_{t+j-1}$,
the truth of $t+j-1 \in J$ and whether $t+j-1 < a$ or $a \leq t+j-1 < b$
is the same for all $j = 1, \dots, k$.
If this were not the case, and there were a time $t+h$ in which one of such relation changes,
then $t+h$ would be a bound of some operator, and it would be included in $K(u_{p_i})$ from the definition of the \textsf{JUMP} rule,
and thus $t(u_{p_{i+1}}) = t+h$.

Let $\suntil[J]{[a,b]}{\varphi_1}{\varphi_2} \in \overline{\Delta}_{t+j-1}$ with $t+j-1 < a$
(the argument for the release operator is analogous).
If $t+j-1 \not\in J$, then $\suntil[J]{[a,b]}{\varphi_1}{\varphi_2} \in \delta_{t+j}$ and is also in $\Delta_{t+j}$,
which satisfies rule \ref{item:model-temp-lta} of Def.~\ref{def:model}.
If $t+j-1 \in J$, then by construction
$\suntil[J]{[a-j+1,b-j+1]}{\varphi_1}{\varphi_2} \in \Gamma(u_{p_i})$,
and by Lemma~\ref{lemma:operator-ladder} also
$\suntil[J]{[a-j,b-j]}{\varphi_1}{\varphi_2} \in \Gamma(u_{p_i})$
(or $\msuntil[J]{[a-j,b-j]}{\varphi_1}{\varphi_2} \in \Gamma(u_{p_i})$ if $t+j = a$),
and hence $\suntil[J]{[a,b]}{\varphi_1}{\varphi_2} \in \delta_{t+j}$
(and $\suntil{[a,b]}{\varphi_1}{\varphi_2} \in \Delta_{t+j}$).

Let $a \leq t+j-1 < b$.
We have $t+j-1 \in J$ iff $t(u_{p_i}) \in J$.
Hence, if $t+j-1 \not\in J$, then $\suntil[J]{[a,b]}{\varphi_1}{\varphi_2} \in \delta_{t+j}$
because $\msuntil[J]{[a,b]}{\varphi_1}{\varphi_2} \in \Gamma(u_{p_i})$.
Since we always used $\Gamma_\phi(u_2)$ when applying the expansion rules to $\delta_{t+j-1}$ to obtain $\Delta_{t+j-1}$,
we have $\exp^{t+j-1}(\varphi_1) \in \Delta_{t+j-1}$, which satisfies rule \ref{item:model-suntil-atb} of Def.~\ref{def:model}.
If, on the other hand, $t+j-1 \in J$, then by Lemma~\ref{lemma:operator-ladder}
we have $\msuntil[J]{[a-j,b-j]}{\varphi_1}{\varphi_2} \in \Gamma(u_{p_i})$,
and so $\suntil[J]{[a,b]}{\varphi_1}{\varphi_2} \in \overline{\Delta}_{t+j}$.
Again, $\exp^{t+j-1}(\varphi_1) \in \Delta_{t+j-1}$, which satisfies rule \ref{item:model-suntil-atb}.

If $t+j-1 = b$, then we must have $t+j-1 \in J$ (or the \textsf{JUMP} rule would not apply).
In this case, while building $\Delta_{t+j-1}$ the $\Gamma_\phi(u_1)$ part of the expansion rule was chosen,
and $\exp^{t+j-1}(\varphi_2) \in \Delta_{t+j-1}$,
which satisfies rule \ref{item:model-suntil-eqb} of Def.~\ref{def:model}.
\end{proof}

\subsection{Completeness}
The tableau method is complete: we prove that if a formula is satisfiable,
then the tableau built for the formula has at least one accepting branch.
We first prove the completeness of the basic tableau, built by only using the \textsf{STEP} rule,
and then show that if a formula has an accepting branch in it,
it has one also in a tableau using the \textsf{JUMP} rule.

Lemma~\ref{lemma:model} states that if there exists a signal that satisfies $\phi$,
then also a model for it exists.
We first show that, given a model for $\phi$,
we can descend the tableau built only by the \textsf{STEP} rule and find an accepting branch for $\phi$.
\begin{lemma}
\label{lemma:step-complete}
If a formula $\phi$ is satisfiable, then its basic tableau has an accepting branch.
\end{lemma}
\begin{proof}
Let $\mathbf{\Delta} = \Delta_0 \Delta_1 \dots \Delta_m$ be a model for $\phi$.
We inductively build a node sequence $\mathbf{u} = u_0, \dots, u_n$ such that $u_n$ is a leaf
and for all $i = 0, \dots, n$, $u_i$ is not rejected, and $\unmark(\Gamma(u_i)) \subseteq \Delta_{t(u_i)}$.

The base case is trivially proved, because $\Gamma(u_0) = \set{\phi} \subseteq \Delta_0$.

Let $i > 0$ with $\Gamma(u_i) \subseteq \Delta_{t(u_i)}$.
We choose $u_{i+1}$ as follows.
\begin{itemize}
\item If $u_i$ is not a poised node, then an expansion rule from Table~\ref{tab:expansion-rules} has been applied to it.
    If it is the rule for the $\land$ operator, $u_i$ has only one child, which we set to $u_{i+1}$.
    Otherwise, $u_i$ has two children $u'$ and $u''$ in which one formula $\psi$ is replaced by
    the contents of either $\Gamma_\psi(u_1)$ or $\Gamma_\psi(u_2)$ from Table~\ref{tab:expansion-rules}.
    In both cases, we have either $\unmark(\Gamma(u')) \subseteq \Delta_{t(u_i)}$ or $\unmark(\Gamma(u'')) \subseteq \Delta_{t(u_i)}$
    according to Def.~\ref{def:model}, and we choose $u_{i+1}$ accordingly.
    In particular, if $\psi$ is a temporal operator, the expansion rules are only applied if
    $t(u_i)$ falls within its interval, satisfying rules \ref{item:model-suntil-atb}--\ref{item:model-srelease-eqb}
    of Def.~\ref{def:model}.

    Expansion rules may cause the appearance of terms in $\Gamma(u_{i+1})$.
    By Def.~\ref{def:model}, there exists a signal that satisfies them,
    hence rules \textsf{FALSE} and \textsf{LOCALLY-UNSAT} do not reject $u_{i+1}$.
    Moreover, by rule \ref{item:model-suntil-eqb} all until operators are satisfied in $\mathbf{\Delta}$,
    hence $\Delta_{t(u_i)}$ does not contain any formula $\suntil{[a,b]}{\varphi_1}{\varphi_2}$ with $b = t(u_i)$,
    and since $\Gamma(u_{i+1}) \subseteq \Delta_{t(u_{i+1})}$, the \textsf{UNTIL} rule cannot reject $u_{i+1}$.

\item If $u_i$ is a poised node containing temporal operators,
    we set $u_{i+1}$ to the only child of $u_i$ generated by the \textsf{STEP} rule.
    The \textsf{STEP} rule does not introduce terms into the new node's label,
    so rules \textsf{FALSE} and \textsf{LOCALLY-UNSAT} cannot reject it.
    Again, rule \ref{item:model-suntil-eqb} of Def.~\ref{def:model}
    prevents the \textsf{UNTIL} rule from rejecting $u_{i+1}$.
\end{itemize}
In both cases, $u_{i+1}$ may contain no temporal operators.
If no more expansion rules can be applied to it, $u_{i+1}$ is accepted by the \textsf{EMPTY} rule, and $i + 1 = n$.
Since, by Def.~\ref{def:time-horizon}, $\phi$ has a finite time horizon $\horiz(\phi)$,
an accepted node $u_n$ will eventually be reached such that $t(u_n) \leq \horiz(\phi)$.
\end{proof}

\begin{lemma}
\label{lemma:jump-complete}
If the basic tableau for a formula $\phi$ has an accepted branch,
then the tableau built by using the \textsf{JUMP} rule also has an accepted branch.
\end{lemma}
\begin{proof}
Let $\mathbf{u} = u_0 \dots u_n$ be an accepted branch in the basic tableau built for a formula $\phi$,
and $\textbf{p} = p_0 \dots p_m$ be the sequence of indices
such that $u_{p_0} \dots u_{p_m}$ is the sequence of poised nodes in $\mathbf{u}$.
Let $u_{p_i}$ be a poised node in $\mathbf{u}$ such that the \textsf{JUMP} rule applies to $u_{p_i}$ but not to $u_{p_{i-1}}$.
The tableau built by using the \textsf{JUMP} rule contains a branch in which $u_{p_i}$ is followed by $u'_{p_i}$,
the node generated by the \textsf{JUMP} rule.

If for each application of the expansion rules in all nodes $v \in \mathbf{u}$ such that $t(u_i) \leq t(v) \leq t(u'_{p_i})$
the $\Gamma(u_2)$ part of Table~\ref{tab:expansion-rules} had been chosen for temporal operators
except when $t(v)$ is equal to their interval's upper bound,
and the same choice as was made between $u_{p_i-1}$ and $u_{p_i}$ was made for the $\lor$ operator,
then we would have $u_{p_{i+k}} = u'_{p_i}$, where $k$ is the jump length $t(u'_{p_i}) - t(u_{p_i})$.

This can be proved by noting that, due to the definition of $k$,
for each temporal operator $\psi = \anybinop[J]{[a,b]}{\varphi_1}{\varphi_2} \in \Gamma(u_{p_i})$
such that $t(u_{p_i}) \not\in J$, either $t(v) < a$ or $a \leq t(v) < b$
for all $v \in \mathbf{u}$ such that $t(u_i) \leq t(v) < t(u'_{p_i})$.
Thus, if $t(u_i) < a$, trivially $\psi \in \Gamma(u_{p_{i+j}})$ for all $j = 0, \dots, k$,
and the same can be proved by induction if $a \leq t(u_i) < b$ by showing
that the $\Gamma_\psi(u_2)$ part of Table~\ref{tab:expansion-rules} can always be chosen when expanding $\psi$.
For formulas $\psi$ such that $t(u_{p_i}) \in J$,
the claim that $\psi \in \Gamma(u_{p_{i+k}})$ follows from Lemma~\ref{lemma:operator-ladder}.

If, indeed, $\mathbf{u}$ is such that $u_{p_{i+k}} = u'_{p_i}$, then the claim follows trivially.
Otherwise, there is a set $\Psi$ of temporal operators $\psi \in \Gamma(u_{p_i})$
such that for some $j \leq k$ we have $\psi \not\in \Gamma(u_{p_{i+j}})$.

Let w.l.o.g.\ $\psi = \suntil[J]{[a,b]}{\varphi_1}{\varphi_2}$ with $t(u_{p_i}) \not\in J$
(the argument for the release operator is analogous).
If $\psi \not\in \Gamma(u_{p_{i+j}})$ then we know that in some node between $u_{p_{i+j-1}}$ and $u_{p_{i+j}}$
the $\Gamma_\psi(u_1)$ part of the until expansion rule in Table~\ref{tab:expansion-rules} was applied,
marking the satisfaction of $\psi$.
We argue that there exists an accepting branch $\mathbf{v}$ in the basic tableau
such that $v_h = u_h$ for all $h = 0, \dots, p_{i-1}$, and $v_{p_i}$ is a poised node such that
$\Gamma(v_{p_i}) \subseteq \Gamma(u_{p_i}) \setminus \set{\psi}$ and the \textsf{JUMP} rule can be applied to it.
Such branch is the one obtained by repeating all applications of the expansion rules as in $\mathbf{u}$,
except for the one of $\psi$ in the nodes between $v_{p_{i-1}+1}$ and $v_{p_i}$, for which we choose $\Gamma_\psi(u_1)$.
This choice does not prevent $\mathbf{v}$ from reaching an accepted node,
because it does not cause the nodes up to $v_{p_i}$ to be rejected.
Indeed, the \textsf{UNTIL} rejection rule cannot be triggered because $\Gamma_\psi(u_1)$ does not include $\psi$.
The \textsf{FALSE} and \textsf{LOCALLY-UNSAT} rules can only be triggered if $v_{p_i}$ contains some term that is not in $u_{p_{i+j-1}}$.
Terms, however, only come from the expansion of temporal operators, which means that $v_{p_i}$ contains some temporal operator that $u_{p_{i+j-1}}$ does not.
Let $\theta = \anybinop[H]{[c,d]}{\theta_1}{\theta_2}$ be such an operator.
We have $t(v_{p_i}) \not\in H$, because $t(v_{p_i}) \in H$ and $\theta \in \Gamma(v_{p_i})$ but $\theta \in \Gamma(v_{p_i})$
would contradict Lemma~\ref{lemma:operator-ladder} (note that $\theta$ could also appear marked).
If $\theta$ disappears because the $\Gamma_\theta(u_1)$ part of Table~\ref{tab:expansion-rules} is applied, we can repeat the reasoning with $\psi = \theta$.
If $\theta$ disappears because $t(v_{p_i}) < d < t(u_{p_{i+j}})$,
then we have $d \in K$, and hence $k = d$, which contradicts $d < t(u_{p_{i+j}})$.
If $t(v_{p_i}) = t(u_{p_i}) = d$, then the \textsf{JUMP} rule does not apply to $u_{p_i}$,
which also leads to a contradiction.

If we have $\psi$ with $t(u_{p_i}) \in J$ and $u_{p_{i+j}}$ is the first node such that $\psi \not\in \Gamma(u_{p_{i+j}})$,
then by a similar reasoning we can argue that
the $\Gamma_\psi(u_1)$ part of the until expansion rule can also be applied to
$\psi' = \suntil[J]{[a-j,b-j]}{\varphi_1}{\varphi_2}$ between $u_{p_{i-1}+1}$ and $u_{p_i}$,
and another accepting branch $\mathbf{v}$ exists equal to $\mathbf{u}$ except that
the $\Gamma_\psi(u_1)$ part of the expansion rule is always chosen for periodic occurrences of $\psi$
between the first node $v$ such that $t(v) = t(u_{p_i})$ and the last node $v'$ such that $t(v') = t(u_{p_{i+k}})$.

By taking branch $\mathbf{v}$ as described above for all formulas in $\Phi$, we can always find an accepting branch in the tableau built without using the \textsf{JUMP} rule.
\end{proof}

From Lemmas~\ref{lemma:step-complete} and \ref{lemma:jump-complete} follows
\begin{theorem}[Completeness]
If a formula is satisfiable, then the tableau rooted in it has an accepted branch.
\end{theorem}

\subsection{Complexity}
\label{sec:complexity}

Consider formula $\glob{[0, k_1]} \even{[0, k_2]} p \land \glob{[0, k_1 + k_2]} \neg p$.
An accepted branch in the tableau rooted in it must have at least $k_1 + k_2$ poised nodes,
because the \textsf{JUMP} rule can only be applied to nodes $u$ with $t(u) \geq k_1 + k_2$.
The tableau can be implemented as a depth-first search
that keeps in memory only the current branch,
using memory at most exponential in the binary representation of $k_1$ and $k_2$.
Local consistency checks can be performed in polynomial time by a linear real arithmetic solver~\cite{Khachiyan80}, so they do not dominate complexity.

A lower bound for STL satisfiability follows from the satisfiability problem
of MTL on discrete time domains, which was proven EXPSPACE-hard
by a reduction from EXPSPACE-bounded Turing machines \cite{AlurH93}.
Thus, discrete-time STL satisfiability is EXPSPACE-complete.

\section{Heuristics and Optimizations}
\label{sec:heur}

Our implementation includes heuristics and other optimizations that speed up tableau construction.
We describe them in this section, and empirically test their effectiveness in Sec.~\ref{sec:ablation}.

\subsection{Early termination}
\label{sec:early-termination}
Our implementation of the tableau stops generating the tree as soon as it finds an accepted branch,
instead of building the complete tableau.


\subsection{Syntactic Manipulations and other operator-specific rules}
\label{sec:syntactic-manip}

\paragraph{Specialized Rules}
We already introduced special rules for the $\globop$ and $\evenop$ operators in Sec.~\ref{sec:tableau}.
We also use a special rule for implication: instead of using the standard equivalence $\varphi_1 \implies \varphi_2 \equiv \neg \varphi_1 \lor \varphi_2$,
we add to category P of Table~\ref{tab:expansion-rules} a new expansion rule that generates two nodes in which $\varphi_1 \implies \varphi_2$ is respectively replaced with $\neg \varphi_1$ and $\varphi_1, \varphi_2$.
We observed empirically that asserting the antecedent together with the subsequent often improves performances
by causing the rejection of redundant subtrees, despite the resulting node being larger.

Nested temporal operators can make the tableau exponentially large.
To mitigate this problem, we add an expansion rule that replaces $\glob{I} \even{I'}$ formulas according to the equivalence
\[
\glob{[a,b]} \even{[c,d]} \varphi \equiv
(\textcolor{blue}{\even{[a+c+1,a+d]} \varphi} \lor (\textcolor{purple}{\glob{[a+c,a+c]} \varphi \land \glob{[a+d+1,a+d+1]} \varphi})) \land \glob{[a+2,b]} \even{[c,d]} \varphi
\]
The right-hand-side formula ``unrolls'' $\glob{[a,b]}$ by two time units:
$\glob{[a,a+1]} \even{[c,d]} \varphi$ can be decomposed in
$\even{[a+c,a+d]} \varphi \land \even{[a+1+c,a+1+d]} \varphi$,
which is true if, either, $\varphi$ holds anywhere in the overlapping part of the two $\evenop$ intervals
(\textcolor{blue}{blue formula}),
or if it holds at both times $a+c$ and $a+d+1$ (\textcolor{purple}{purple formula}).

\paragraph{Removing redundant formulas}
We use the following equivalences to merge operators with contiguous or overlapping intervals and reduce node size:
\begin{align*}
&\glob{[a,b]} \varphi \land \glob{[c,d]} \varphi \iff \glob{[a,d]} \varphi && \text{if } a \leq c \leq b \leq d \\
&\even{[a,b]} \varphi \land \even{[c,d]} \varphi \iff \even{[c,d]} \varphi && \text{if } c \geq a \land d \leq b.
\end{align*}

\paragraph{Shifting time intervals of nested operators}
We generalize the first step \eqref{eq:railroad-shift-interval} of the inconsistency proof for the example in Sec.~\ref{sec:intro}.
For any operator $\psi = \anybinop{[a,b]}{\varphi_1}{\varphi_2}$, with $\mathsf{B} \in \set{\suntilop, \sreleaseop}$,
such that $\Phi = \pcl(\varphi_1) \cup \pcl(\varphi_2)$ contains only temporal operators,
we compute $c_{\min}$ as the minimum lower bound of all operators in $\Phi$,
and replace $\psi$ with $\anybinop{[a+c_{\min},b+c_{\min}]}{\varphi'_1}{\varphi'_2}$,
where $\varphi'_1$ and $\varphi'_2$ are obtained by shifting all time intervals of operators in $\Phi$
back by $c_{\min}$.
We apply this transformation recursively, starting with innermost formulas and proceeding outwards.

\subsection{Early consistency check}
\label{sec:early-consistency-check}
Rather than performing the local consistency check (i.e., the \textsf{LOCALLY-UNSAT} rule)
only on poised nodes, we do it before each decomposition step, to avoid fully decomposing nodes that already contains inconsistent atomic expressions.

\subsection{Memoization}
\label{sec:memo}
Nodes in the tableau can often be redundant, especially when generated by nested temporal operators.
To speed up tableau construction, we memorize rejected nodes,
and after each application of the \textsf{STEP} and \textsf{JUMP} rules we check if the resulting node
implies one of the previously rejected nodes.
If this is the case, the new node can be safely rejected without building the subtree rooted in it.
Since a proper implication check would be too expensive, we use a weaker one that is quicker,
and exploits the following relations:
\begin{align*}
\glob{[a,b]} \varphi &\implies \glob{[c,d]} \varphi && \text{if } a \leq c \leq d \leq b,
&
\even{[a,b]} \varphi &\implies \even{[c,d]} \varphi && \text{if } c \leq a \leq b \leq d \\
\srelease{[a,b]}{\varphi_1}{\varphi_2} &\implies \srelease{[c,d]}{\varphi_1}{\varphi_2} && \text{if } a = c \leq d \leq b,
&
\suntil{[a,b]}{\varphi_1}{\varphi_2} &\implies \suntil{[c,d]}{\varphi_1}{\varphi_2} && \text{if } c = a \leq b \leq d
\end{align*}

\subsection{Easy operators first}
\label{sec:easy-first}

Checking nested temporal operators can be very slow, especially when the formula is not satisfiable and the tableau must be generated completely.
However, contradiction can often be assessed by checking just non-nested operators that are extracted by the expansion rules.
For instance, consider a tableau node labeled with $\glob{[0,10]} \even{[30,40]} x > 0, \glob{[30,50]} x \leq 0$.
Since the \textsf{JUMP} rule can only be used after time 40,
we must develop the tableau by only applying expansion and the \textsf{STEP} rules until then.
The $\glob{I} \even{I'}$ formula causes the tableau size to grow exponentially,
making its construction infeasible.
However, the first application of the $\globop$ expansion rule creates a node labeled with
$\mglob{[0,10]} \even{[30,40]} x > 0, \textcolor{blue}{\glob{[30,50]} x \leq 0, \even{[30,40]} x > 0}$:
we can infer unsatisfiability just from the last two formulas,
without needing to develop the entire tableau for the original formula.

\begin{wrapfigure}[18]{r}
\footnotesize
\begin{forest}
  for tree={
    myleaf/.style={label=below:{\strut#1}}
  },
  [{$\glob{[0,10]} \even{[30,40]} p, \glob{[30,50]} \neg p \mid \mathbf{0}$}
    [{$\mglob{[0,10]} \even{[30,40]} p, \glob{[30,50]} \neg p, \even{[30,40]} p \mid \mathbf{0}$},
     edge label={node[midway,left,font=\scriptsize\sffamily]{G}},
      [{\textcolor{blue}{$\glob{[30,50]} \neg p, \even{[30,40]} p \mid \mathbf{1}$}},
       edge label={node[midway,left,xshift=20pt,font=\scriptsize\sffamily]{STEP}},
        [{$\glob{[30,50]} \neg p, \even{[30,40]} p \mid \mathbf{30}$},
         edge label={node[midway,left,font=\scriptsize\sffamily]{JUMP}},
          [{$\dots$}, myleaf={\xmark\ \textsf{LOCALLY-UNSAT}}]
        ]
      ],
      [{$\left.\begin{aligned}&\mglob{[0,10]} \even{[30,40]} p,\\&\glob{[30,50]} \neg p, \even{[30,40]} p\end{aligned} \right| \mathbf{1}$}]
    ],
  ]
\end{forest}
\caption{The \textcolor{blue}{blue} node and its sibling are treated as AND nodes.
  The \textsf{JUMP} rule can be applied to the purple node,
  and the tableau quickly shows that its label is unsatisfiable,
  and so must be its sibling. We have $p = x > 0$.}
\label{fig:easy-first-example}
\end{wrapfigure}
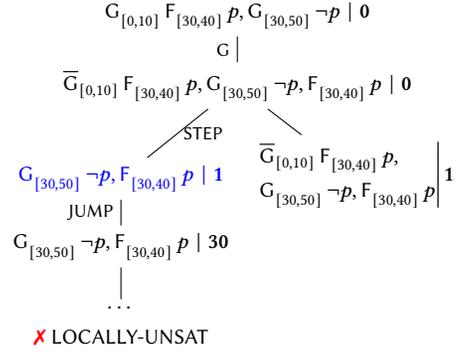
Thus, we change the  \textsf{STEP} and \textsf{JUMP} rules to return two children:
one containing only non-nested operators, and one as defined in Sec.~\ref{sec:tableau}.
We treat the poised node as an AND-node in an AND-OR tree:
we check its first child---the one without nested operators---first,
and the second only if the first leads to an accepted branch
(cf.\ Fig.~\ref{fig:easy-first-example}).

This feature, combined with memoization, can prune entire subtrees,
proving very effective in speeding up the tableau for unsatisfiable formulas.

\section{Experimental Evaluation}
\label{sec:exp}

We implemented the tableau for STL in Python, using the Z3 SMT solver~\cite{Z3} to check satisfiability of linear constraints.
The experimental evaluation comprises different benchmark sets taken from literature and industry.
We carried out all the experiments on a machine equipped with a 4.5GHz AMD CPU and 64~GB of RAM running Ubuntu 24.04.
Our tool and benchmarks are available at~\cite{stltree}.

\subsection{STL Benchmarks}
\label{sec:stl-benchmarks}

We compared our implementation of the tableau with the SMT encoding by \citet{RamanDMMSS14},
which we re-implemented without the parts devoted to controller synthesis,
so that it only performs satisfiability checking.
Our implementation discharges constraints on Z3.

This benchmark suite consists of several sets of STL requirements.
The first five sets are adapted from \cite{BaeLee19}. 
They consist of five different models (\emph{Car}, \emph{Thermostat}, \emph{Water Tank}, \emph{Railroad}, and \emph{Battery}) described by four STL properties each.
MTL is a set of requirements for an automatic transmission and fault-tolerant fuel control system adapted from \cite{Benchmarks_Temporal_Logic_Requirements}.
PCV is a set of requirements of a powertrain control from \cite{pcv}.
Lastly, CPS requirements, adapted from \cite{BoufaiedJBBP21}, describe a satellite sub-system called \emph{Attitude Determination and Control System}, which is responsible for controlling the orientation of a satellite with respect to a reference point.

Another set of requirements comes from the aerospace domain and was provided by an industrial partner. It consists of 31 item-level software functional requirements taken from a real system; an example is shown in Fig.~\ref{fig:aerospace-reqs}.
The experimental results are shown for $T=1000$.

\begin{stlbox}[
    float,
    title={Example of STL Requirements from an industrial aerospace project.},
    label={fig:aerospace-reqs}
]
\renewcommand{\arraystretch}{1.1} 
\footnotesize
\begin{tabular}{p{13cm}} 
    $\glob{[0,T]} \big( (\mathit{inactive} \wedge (n_s =1)  \wedge (|X_c - X_b |\leq 5) \wedge $\\$ \glob{[0,5]}(\mathit{airspeed} \geq V_\mathit{min}) \wedge \neg X_\mathit{over} \wedge$ $X_\mathit{ActivationRequest}) \big)\implies \even{[1,3]} \mathit{active})$ \\
    This requirement specifies the condition of transition from \emph{inactive} to \emph{active} state. \\
    
    \midrule

    $\glob{[0,T]}(\mathit{inactive} \implies \even{[0,5]} (\mathit{LMT}_\mathit{ar} \wedge a_\mathit{tone}))$ \\
    This requirement describes what should happen while in a \emph{inactive} state. \\

    \midrule

    $\glob{[0,T]}((X_c \geq 0) \wedge (X_c \leq 360))$ \\
    This requirement constrains the ranges of the parameters/variables. \\
\end{tabular}
\end{stlbox}

\begin{table*}
    \centering
    \caption{Results for the STL benchmarks.
    The second column shows the number of requirements contained in each benchmark.
    Columns 3,4,5 resp.\ report maximum nesting depth of temporal operators, Time Horizon (as per \ref{def:time-horizon}), and temporal operators included in the formula.
    Columns 6,7 show solution time (s) and result using the SMT encoding \cite{RamanDMMSS14}, and columns 8,9 for the Tableau.
    $\top$ and $\bot$ mean resp.\ ``satisfiable'' and ``unsatisfiable''.}
    \label{tab:execution_times}
    \footnotesize
    \begin{tabular}{l r r r c r c r c}
        \toprule
        \textbf{Benchmark} & \# & \textbf{max n} & \textbf{TH} & \textbf{TOps} & \textbf{SMT} & \textbf{Res.} & \textbf{Tableau} & \textbf{Res.} \\
        \midrule
        Car \cite{BaeLee19} & 4 & 2 & 100 & G,F,U & 0.213 & $\top$ & \textbf{0.052} & $\top$\\
        Thermostat \cite{BaeLee19} & 4 & 2 & 40 & G,F,U,R & 0.070 & $\top$  & \textbf{0.037} & $\top$  \\
        Watertank \cite{BaeLee19} & 4 & 2 & 50 & G,F & 0.095 & $\top$ & \textbf{0.040} & $\top$\\
        Railroad \cite{BaeLee19} & 4 & 2 & 100 & G,F & \textbf{0.136} & $\bot$ & 0.239 & $\bot$ \\
        Battery \cite{BaeLee19} & 4 & 2 & 64 & G,F,U & 0.114 & $\top$ & \textbf{0.093} & $\top$ \\
        PCV \cite{pcv} & 9 & 2 & 1,250 & G,F & 19.095 & $\top$ & \textbf{1.195} & $\top$\\
        MTL \cite{Benchmarks_Temporal_Logic_Requirements} & 10 & 3 & 2,020 & G,F & 10.932 & $\bot$ & \textbf{0.696} & $\bot$ \\
        CPS \cite{BoufaiedJBBP21} & 37 & 3 & 13,799 & G,F,U & 120.000 & TO & \textbf{1.909} & $\bot$\\
        Aerospace & 31 & 2 & 1,010 & G,F & 14.627 & $\top$ & \textbf{4.859} & $\top$ \\
        \bottomrule
    \end{tabular}
\end{table*}
The experimental results, summarized in Table \ref{tab:execution_times}, show that the tableau outperforms the SMT encoding 9 times out of 10 for non-trivial formulas with nested temporal operators.
The SMT encoding outperforms the tableau only on the \emph{Railroad} benchmark, from which the example in Sec.~\ref{sec:intro} has been taken.
The combination of these two formulas containing nested temporal operators, in fact, represents a worst-case scenario for the tableau, exacerbated by the requirement set being unsatisfiable, which prevents our implementation from stopping early.
The resulting tableau is, however, highly redundant, which allows memoization to prune a large part of the tableau, whose size would otherwise be doubly exponential in formula size.

The experiments show that three out of the ten requirement sets are inconsistent, further highlighting the usefulness of consistency checking tools in requirement engineering.


\subsection{MLTL Benchmarks}
\label{sec:MLTL}

\begin{figure*}
    \centering
    \begin{subfigure}[b]{0.33\textwidth}
        \centering
        \includegraphics[width=\linewidth]{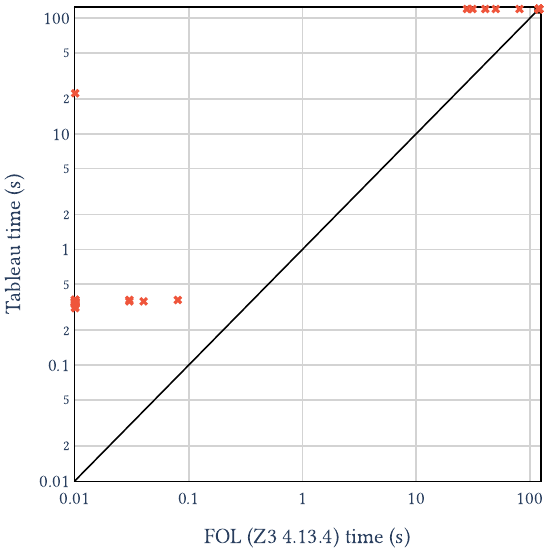}
    \end{subfigure}%
    \begin{subfigure}[b]{0.33\textwidth}
        \centering
        \includegraphics[width=\linewidth]{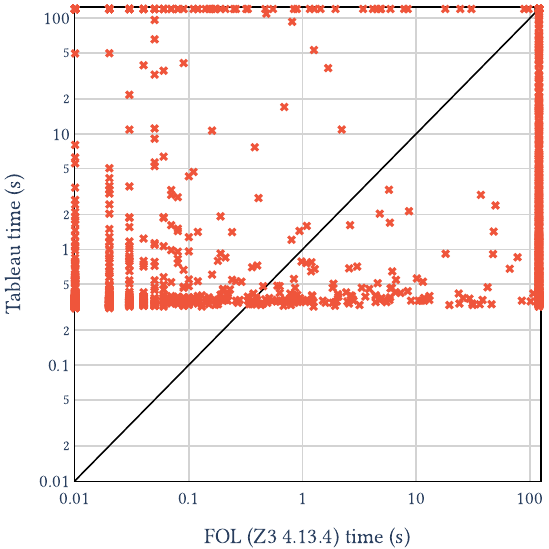}
    \end{subfigure}%
    \begin{subfigure}[b]{0.33\textwidth}
        \centering
        \includegraphics[width=\linewidth]{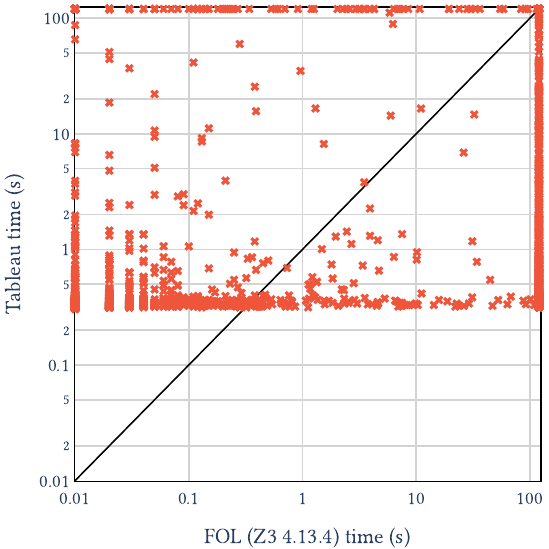}
    \end{subfigure}
    \vspace{1em}

    \begin{subfigure}[b]{0.33\textwidth}
        \centering
        \includegraphics[width=\linewidth]{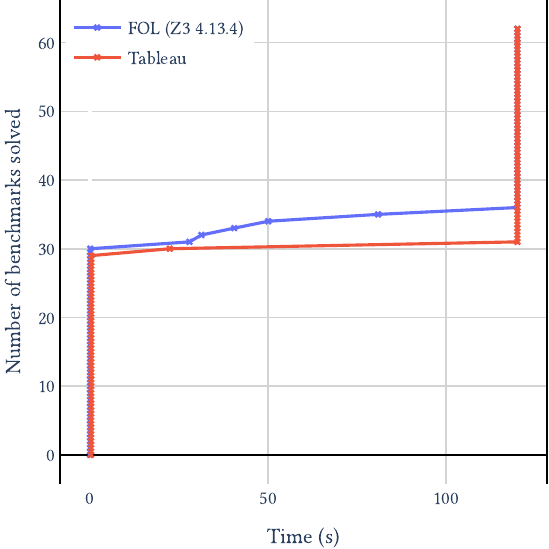}
        \caption{Nasa-Boeing suite}
        \label{fig:nasa}
    \end{subfigure}%
    \begin{subfigure}[b]{0.33\textwidth}
        \centering
        \includegraphics[width=\linewidth]{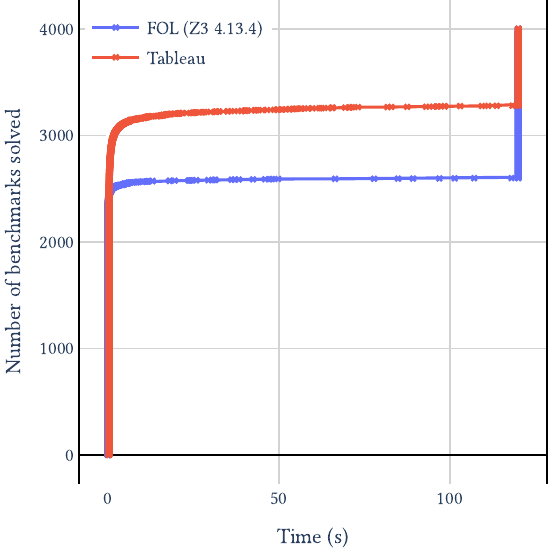}
        \caption{Random suite}
        \label{fig:random}
    \end{subfigure}%
    \begin{subfigure}[b]{0.33\textwidth}
        \centering
        \includegraphics[width=\linewidth]{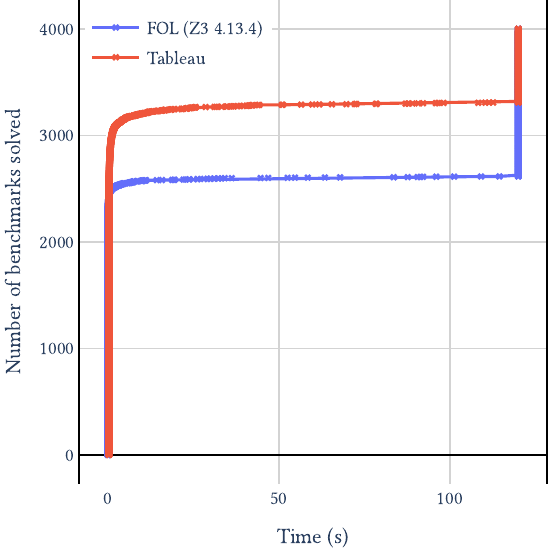}
        \caption{Random0 suite}
        \label{fig:random02}
    \end{subfigure}
    \caption{Scatter (top) and survival (bottom) plots for MLTL benchmarks \cite{LiVR22}.
    FOL is the encoding by~\cite{LiVR22}.
    In scatter plots, benchmarks represented by dots above the identity line are solved faster by FOL, and \emph{vice versa};
    the axes are logarithmic.
    In survival plots, higher is better.}
    \label{fig:combined-plots}
\end{figure*}

MLTL~\cite{ReinbacherRS14} is strictly related to discrete-time STL.
Indeed, it is equivalent to a fragment in which the until operator is replaced with the $\suntilop$ operator defined in Sec.~\ref{sec:stl}, restricted to Boolean signals.
Thus, we compare our tableau with a state-of-the-art approach to check the satisfiability of MLTL requirements \cite{LiVR22}.
\citet{LiVR22} compare several encodings that exploit different techniques to check MLTL satisfiability.
They collected a benchmark suite consisting of 63 requirement sets adapted from NASA and Boeing projects~\cite{DurejaR18}, and two sets of 4000 formulas that were randomly generated based on realistic patterns (\emph{Random} and \emph{Random0}, the latter having temporal intervals all starting from 0).
Their experiments show that their FOL encoding performs best in most benchmarks.
Thus, we compare our tableau with their encoding, running on Z3 4.13.4, on the same benchmarks.
We downloaded implementation and benchmarks from \cite{mltlsat}.

The results are summarized in Fig.~\ref{fig:combined-plots}, showing scatter and survival plots for formulas checked within the 120 s timeout.
Survival plots show how many benchmarks each tool can solve within the time on the $x$-axis,
and are the standard plots used for comparing satisfiability solvers \cite{BrainDG17}.
The tableau always takes a time of at least 0.3 s, due to the overhead of the Python runtime environment.
The encoding by \citet{LiVR22} is able to solve a few more benchmarks on the Nasa-Boeing suite, although the performance of the two techniques is comparable on most benchmarks.
Our tableau drastically outperforms the FOL encoding on randomly generated formulas.
A possible explanation for such differences is that the benchmarks in the NASA-Boeing suite on which the SMT solver performs better contain a large propositional part, on which the SMT solver performs very well, which overcomes the time required to solve temporal constraints.
The random benchmarks, on the other hand, contain several temporal operators with very large intervals, on which our tableau performs better.
The experiments suggest that our tableau could be used together with the FOL encoding by \citet{LiVR22}, since it performs better with different kinds of requirements.

\subsection{Ablation analysis of tableau optimizations}
\label{sec:ablation}

\begin{table}[bt]
    \centering
    \caption{Ablation analysis on STL benchmarks.
    The 2\textsuperscript{nd} column reports the solution times with all optimizations enabled (same as Table~\ref{tab:execution_times}, repeated for comparison);
    all subsequent columns report solution times after disabling the optimization in the column header.
    All times are in seconds, the timeout (TO) is 120~s.}
    \label{tab:ablation-stl}
    \footnotesize
    \begin{tabular}{l r r r r r r r}
        \toprule
        \textbf{Benchmark} & All opts. & \textsf{JUMP} & No G,F rules & Syntactic & Early check & Memoization & Easy ops. \\
        & & §\ref{sec:tableau} & & §\ref{sec:syntactic-manip} &  §\ref{sec:early-consistency-check} & §\ref{sec:memo} & §\ref{sec:easy-first} \\
        \midrule
        Car \cite{BaeLee19} & 0.052 & 0.501 & 0.049 & 0.045 & 0.046 & 0.046 & \textbf{0.025} \\
        Thermostat \cite{BaeLee19} & 0.037 & 0.181 & 0.035 & 0.034 & 0.035 & 0.034 & \textbf{0.024} \\
        Watertank \cite{BaeLee19} &  0.040 & 0.280 & 0.041 & 0.040 & 0.039 & 0.039 & \textbf{0.026} \\
        Railroad \cite{BaeLee19} & \textbf{0.239} & 0.750 & TO & TO & 0.264 & TO & 0.380 \\
        Battery \cite{BaeLee19} & 0.094 & 0.300 & 0.101 & 0.108 & 0.086 & 0.117 & \textbf{0.065} \\
        PCV \cite{pcv} & 1.195 & 93.350 & 1.464 & 1.155 & 1.173 & 1.171 & \textbf{0.531} \\
        MTL \cite{Benchmarks_Temporal_Logic_Requirements} & \textbf{0.697} & 0.745 & TO & 0.849  & 1.209 & 14.461 & TO \\
        CPS \cite{BoufaiedJBBP21} &  \textbf{1.909} & TO & TO & TO & TO & TO & TO \\
        Aerospace & 4.859 & TO & 15.885 & TO & TO & 4.943 & \textbf{3.226} \\
        \bottomrule
    \end{tabular}
\end{table}

We conducted an ablation analysis to assess the impact of the optimizations described in Sec.~\ref{sec:heur}.
The results on the STL benchmarks introduced in Sec.~\ref{sec:stl-benchmarks} are reported in Table~\ref{tab:ablation-stl},
and those for MLTL benchmarks are shown through survival plots in Fig.~\ref{fig:ablation-plots}.

\begin{figure}
    \centering
    \begin{subfigure}[b]{0.33\textwidth}
        \centering
        \includegraphics[width=\linewidth]{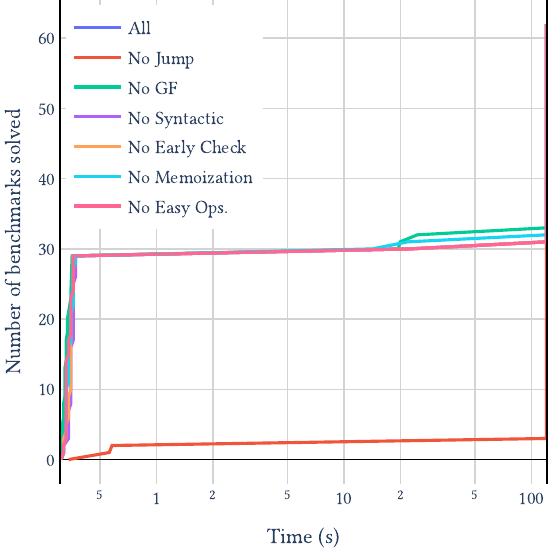}
        \caption{Nasa-Boeing suite}
        \label{fig:nasa2}
    \end{subfigure}%
    \begin{subfigure}[b]{0.33\textwidth}
        \centering
        \includegraphics[width=\linewidth]{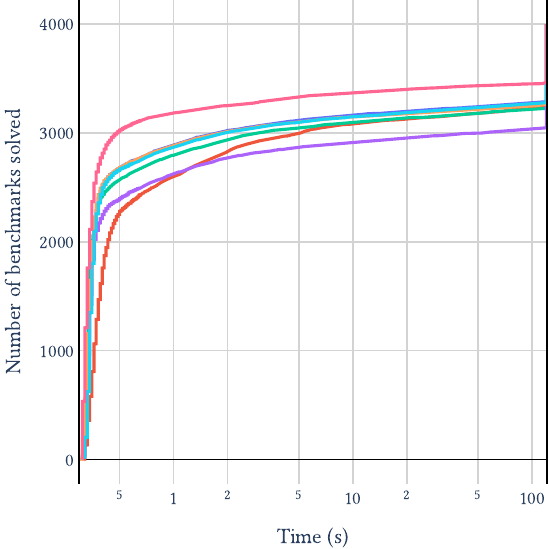}
        \caption{Random suite}
        \label{fig:random2}
    \end{subfigure}%
    \begin{subfigure}[b]{0.33\textwidth}
        \centering
        \includegraphics[width=\linewidth]{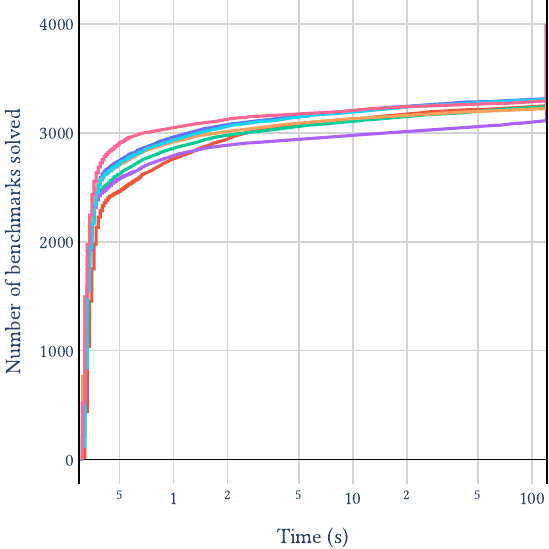}
        \caption{Random0 suite}
        \label{fig:random0}
    \end{subfigure}
    \caption{Survival plots for the ablation analysis on MLTL benchmarks from \cite{LiVR22}.
    Time axes are logarithmic.}
    \label{fig:ablation-plots}
\end{figure}

Railroad, PCV, MTL, CPS, and Aerospace are the benchmarks that show the most significant variations when disabling optimizations,
as well as the Random suites from \cite{LiVR22}.

For Railroad, this phenomenon can be explained by the particular structure of the two contradictory formulas analyzed in Sec.~\ref{sec:intro}.
The first one contains a $\even{}$ nested into a $\glob{}$%
---a worst-case scenario for the tableau---%
which prevents the \textsf{JUMP} rule from significantly shortening tree branches
(cf.\ Sec.~\ref{sec:complexity}).
The second formula contains an implication nested into a $\glob{}$ operator,
which causes branching at every time instant, thus exponentially increasing tableau size.
Since the two formulas are contradictory,
early termination does not apply (Sec.~\ref{sec:early-termination})
and the tableau must be built entirely.

For the remaining benchmarks, the impact of optimizations can be explained by their large time horizon and number of requirements (hence, size).

Disabling the \textsf{JUMP} rule (2\textsuperscript{nd} column)
mostly affects PCV, CPS, and Aerospace, and all MLTL benchmarks.
This is due to the high time horizon of these requirements (cf.\ Table~\ref{tab:execution_times}).
While MTL also has a long time horizon, it contains contradictions that are isolated by the easy-operators-first optimization (Sec.~\ref{sec:easy-first}),
and all branches are pruned after a few applications of the \textsf{STEP} rule.

The main benefit of using the specialized rules for unary operators $\globop$ and $\evenop$
is that the expansion rule for $\globop$ produces only one child node,
while both $\suntilop$ and $\sreleaseop$ expansion rules produce two children,
increasing tree branching.
If the specialized rules for unary operators $\globop$ and $\evenop$ are disabled (3\textsuperscript{rd} column), Railroad, MTL, and CPS are affected the most.
MTL and CPS are the benchmarks with the highest nesting depth, of mostly $\globop$ and $\evenop$ operators,
which causes many temporal operators to extracted at every application of the expansion rules,
significantly increasing tree size.
Furthermore, both external and nested operators in CPS have large time horizons.
Railroad contains requirements that are worst-case scenarios for the tableau,
an issue exacerbated by the higher tree branching due to using $\suntilop$ and $\sreleaseop$ rules
instead of the specialized $\globop$ and $\evenop$ rules.

Disabling optimizations based on syntactic manipulations (Sec.~\ref{sec:syntactic-manip})
affects mostly Railroad, CPS and Aerospace, mainly because they contain many instances of nested $\glob{}\even{}$ operators.
Moreover, shifting time intervals is particularly beneficial for Railroad,
as we also used it in the manual proof of its inconsistency in Sec.~\ref{sec:intro},
and Random MLTL benchmarks due to their large time intervals.

The early consistency check (Sec.~\ref{sec:early-consistency-check}) greatly benefits CPS and Aerospace.
These benchmarks contain integer signals
that have been modeled with constraints of the form $x == 0 \lor x == 1 \lor x==2$
(for a variable $x$ that can only take one value among $0, 1, 2$).
They also contain many assume-guarantee requirements of the form
$\glob{[a,b]} (x == k \implies \varphi_k)$, for $k \in \set{0,1,2}$,
whose conjunction results in many nodes containing clearly contradictory terms of the form $x == k, x == h$ with $k \neq h$.
At the same time, these benchmark contain many $\lor$, $\implies$, and $\even{}$ operators,
whose expansion rules create a large number of branches%
---precisely, a tree portion of size exponential in the number of such operators.
Early consistency checking prunes nodes before expansion rules can be applied,
avoiding the useless creation of large subtrees.

Memoization is very beneficial to Railroad, MTL, and CPS.
This phenomenon is mainly due to the high redundancy of tableau nodes generated by these benchmarks (cf.\ Fig.~\ref{fig:GF-tree}).
In particular, the presence of nested $\glob{}\even{}$ operators causes large redundant subtrees,
that memoization recognizes and prunes effectively.

The optimization that checks easy operators first is very beneficial to MTL and CPS,
because these benchmarks contain contradictions that arise from non-nested operators,
and checking them first allows to skip the analysis of more complex parts of the requirements.
This optimization is, however, detrimental to almost all other benchmarks,
in a dramatic way for Random MLTL benchmarks (Fig.~\ref{fig:ablation-plots}).
Indeed, for formulas that are satisfiable, or for which finding a contradiction requires analyzing deeply nested operators,
checking a node containing only ``easy'' operators does not provide any benefit,
and causes additional overhead.
However, the speedup enjoyed by MTL and CPS justifies the use of this optimization in some cases.
\color{black}

\section{Conclusion}
We propose a tree-shaped, one-pass tableau method to address the satisfiability of a set of discrete-time STL formulas with bounded temporal operators.  Our approach integrates an SMT solver to check the satisfiability of linear constraints within the nodes of the tree.  We prove that our method is sound and complete.  

We evaluated our approach on several benchmarks and compared it against existing methods that encode the specification as SMT and FOL formulas. The experimental results demonstrate the applicability of our method to a large and complex set of requirements, and it outperforms existing approaches in most of the considered scenarios. 

In future work, we plan to improve our method with new heuristics and extend its application to control synthesis problems and automatic example generation.




\begin{acks}
This work was partially funded by the
\grantsponsor{WWTF}{Vienna Science and Technology Fund (WWTF)}{https://wwtf.at/}
grant
\grantnum{WWTF}{ICT22-023 (TAIGER)}
and by the 
\grantsponsor{EC}{European Commission}{https://commission.europa.eu/}
in the Horizon Europe research and innovation programme
\includegraphics[width=1em]{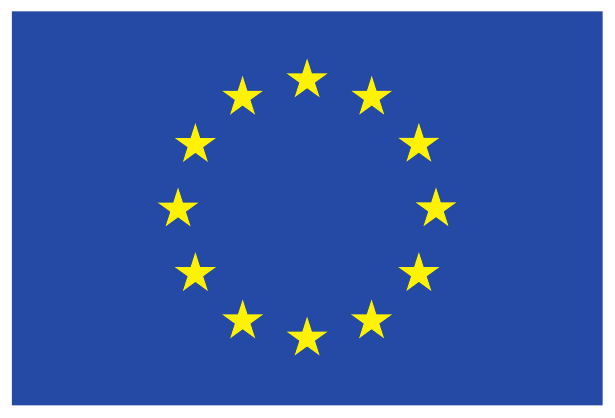}
under grant agreements
\grantnum{EC}{No.\ 101107303 (MSCA Postdoctoral Fellowship CORPORA)},
\grantnum{EC}{No.\ 101160022 (VASSAL)},
\grantnum{EC}{No.\ 101212818 (RobustifAI)}.
\end{acks}

\bibliographystyle{ACM-Reference-Format}
\bibliography{bib}

\end{document}